\documentclass{amsart}
\usepackage{amsmath}
\usepackage{amssymb}
\usepackage{amsfonts}
\usepackage{graphicx}
\usepackage{texdraw}
\usepackage{graphpap}
\usepackage{wasysym}
\usepackage{enumitem}
\usepackage{color}
\usepackage[OT2,T1]{fontenc}

\theoremstyle{plain}
\newtheorem{thm}{Theorem}[section]
\newtheorem{mth}{Theorem}

\newtheorem{mlem}{Lemma}

\newtheorem{lem}[thm]{Lemma}
\newtheorem*{lemm}{Lemma}

\theoremstyle{remark}

\newtheorem*{rem}{Remark}

\newtheorem*{ex}{Example}

\numberwithin{equation}{section}

\newcommand{\bfR}{\mathbf{R}}

\newcommand{\meso}{r}

\newcommand{\calH}{{\mathcal H}}

\newcommand{\compset}{V}
\newcommand{\calP}{{\mathcal P}}
\newcommand{\calM}{{\mathcal M}}

\newcommand{\calW}{{\mathcal W}}

\newcommand{\bfK}{{\mathbf K}}
\newcommand{\bfL}{{\mathbf L}}
\newcommand{\bfk}{{\mathbf k}}
\newcommand\modul{\tau_0}

\newcommand{\Int}{\operatorname{Int}}

\newcommand{\C}{{\mathbb C}}
\newcommand{\E}{{\mathbf E}}

\newcommand{\const}{\mathrm{const.}}
\newcommand{\eps}{{\varepsilon}}

\newcommand{\re}{\operatorname{Re}}
\newcommand{\im}{\operatorname{Im}}

\newcommand{\Poly}{\operatorname{Pol}}

\newcommand{\Prob}{{\mathbf{P}}}

\renewcommand{\d}{{\partial}}
\newcommand{\dbar}{\bar{\partial}}
\newcommand{\1}{\chi}

\newcommand{\dist}{\operatorname{dist}}
\newcommand{\supp}{\operatorname{supp}}

\newcommand{\Lap}{\Delta}

\def\labs{\left |}
\def\rabs{\right |}
\def\babs#1{\labs {#1} \rabs}

\begin{document}

\title{On bulk singularities in the random normal matrix model}



\author{Yacin Ameur}

\address{Yacin Ameur\\
Department of Mathematics\\
Faculty of Science\\
Lund University\\
P.O. BOX 118\\
221 00 Lund\\
Sweden}

\email{Yacin.Ameur@maths.lth.se}

\author{Seong-Mi Seo}

\address{Seong-Mi Seo\\Department of Mathematical Sciences\\Seoul National University\\Seoul\\ 151-747\\ Republic of Korea}

\email{ssm0112@snu.ac.kr}

\thanks{Seong-Mi Seo was supported by Samsung Science and Technology Foundation, SSTF-BA1401-01.}

\begin{abstract} We extend the method of rescaled Ward identities from \cite{AKM} to study the distribution of eigenvalues close to a bulk singularity, i.e. a point in the interior of the droplet where the density of the classical equilibrium measure vanishes. We prove results to the effect that a certain "dominant part'' of the Taylor expansion determines the microscopic properties near a bulk singularity. A description of the distribution is given in terms of a special entire function, which depends on the nature of the singularity (a Mittag-Leffler function in the case of a rotationally symmetric singularity).
\end{abstract}
\maketitle

Consider a system $\{\zeta_j\}_1^n$ of identical point-charges in the complex plane in the presence of an external field $nQ$, where $Q$ is a suitable function.
The system is assumed to be picked randomly with respect to the
Boltzmann--Gibbs probability law at inverse temperature $\beta=1$,
\begin{equation}\label{bglaw}d\Prob_n(\zeta)=\frac 1 {Z_n}e^{-H_n(\zeta)}\, d^{\,2n}\zeta,\end{equation}
where $H_n$ is the weighted energy of the system,
\begin{equation}\label{mini}H_n(\zeta_1,\ldots,\zeta_n)=\sum_{j\ne k}\log\frac 1 {\babs{\,\zeta_j-\zeta_k\,}}+n\sum_{j=1}^n Q(\zeta_j).\end{equation}
The constant $Z_n$ in \eqref{bglaw} is chosen so that the total mass
is $1$.

It is well-known that (with natural conditions on $Q$) the normalized counting measures $\mu_n=\frac 1 n\sum_{j=1}^n\delta_{\zeta_j}$ converge
to \textit{Frostman's equilibrium measure} as $n\to\infty$. This is a probability measure of the form
\begin{equation}\label{frost}d\sigma(\zeta)=\chi_S(\zeta)\,\Lap Q(\zeta)\, dA(\zeta)\end{equation}
where $\chi_S$ is the indicator function of a certain compact set
$S$ called the \textit{droplet}.

We necessarily have $\Lap Q\ge 0$ on $S$.
In the papers \cite{AKM,AKMW}, the method of rescaled Ward identities was introduced and applied to study microscopic properties of the system $\{\zeta_j\}_1^n$ close to a (moving) point $p\in S$. The situation in those papers is however restricted by the condition that the point $p$ be "regular'' in the sense that $\Lap Q(p)\ge\const>0$. In this note, we extend the method to allow for a "bulk singularity'', i.e. an isolated point $p$ in the interior of $S$ at which $\Lap Q=0$.

In general, a bulk singularity tends to repel particles away, which means that one must use a relatively coarse scale in order to capture the relevant structure.
 We prove results to the effect that (in many cases) the dominant terms in the Taylor expansion
of $\Lap Q$ about $p$ determines the microscopic properties of the system in the vicinity of $p$. Our characterization uses the Bergman kernel for a certain space of entire functions, associated with these dominant terms.
In particular, we obtain quite different distributions depending on the degree to which $\Lap Q$ vanishes at $p$.

\begin{rem} It is well-known that the particles $\{\zeta_j\}_1^n$ can be identified with eigenvalues of random normal matrices with a suitable weighted distribution.
The details of this identification are not important for the present investigation. However, following tradition, we shall sometimes speak of a "configuration of random eigenvalues'' instead of a "particle-system''.
\end{rem}

\begin{rem}The meaning of the convergence $\mu_n\to\sigma$ is that
$\E_n[\mu_n(f)]\to \sigma(f)$ as $n\to\infty$ where $f$ is a suitable test-function, e.g. in the Sobolev space
$W^{1,2}(\C)$, where $\E_n$ is expectation with respect to \eqref{bglaw}. In fact, more can be said, see \cite{AM}.
\end{rem}

\subsection*{Notation} We write $\Lap=\d\dbar$ for $1/4$ of the usual Laplacian, and $dA$ for $1/\pi$ times Lebesgue measure on the plane $\C$. Here $\d=\frac 1 2(\d/\d x-i\d/\d y)$ and $\dbar=\frac 1 2(\d/\d x+i\d/\d y)$ are the usual complex derivatives. We write $\bar{z}$ (or occasionally $z^*$) for the complex conjugate of a number $z$.
 A continuous function $h(z,w)$ will be called \textit{Hermitian} if $h(z,w)=h(w,z)^*$.
$h$ is called \textit{Hermitian-analytic} (\textit{Hermitian-entire}) if $h$ is Hermitian and analytic (entire) in $z$ and $\bar{w}$. A Hermitian function $c(z,w)$ is called a \textit{cocycle} if there is a unimodular function $g$ such that $c(z,w)=g(z)\bar{g}(w)$, where for functions we use the notation $\bar{f}(z)=f(z)^*$. We write $D(p,r)$ for the open disk with center $p$ and radius $r$.

\section{Introduction; Main Results}

\subsection{Potential and equilibrium measure} The function $Q$ is usually called the "external potential''. This function is assumed to be lower semi-continuous and real-valued, except that it may assume the value $+\infty$ in portions of the plane. We also assume: (i) the set $\Sigma_0=\{Q<\infty\}$ has dense interior, (ii) $Q$ is real-analytic in $\Int\Sigma_0$, and (iii) $Q$ satisfies the growth condition
\begin{equation}\label{eq:grow}\liminf_{\zeta\to\infty}\frac {Q(\zeta)}{\log\babs{\,\zeta\,}^{\,2}}>1.\end{equation}

For a suitable measure on $\C$, we define its
$Q$-\textit{energy} by
$$I_Q[\mu]=\iint_{\C^2}\log\frac 1 {\babs{\,\zeta-\eta\,}}\, d\mu(\zeta)d\mu(\eta)+\int_\C Q\, d\mu.$$
The \textit{equilibrium measure} $\sigma=\sigma_Q$ is defined as the probability measure which minimizes $I_Q[\mu]$ over all compactly supported Borel probability measures $\mu$. Existence and uniqueness of such a minimizer is well-known, see e.g. \cite{ST} where also the explicit expression \eqref{frost} is derived, with $S=\supp\sigma$.

\subsection{Rescaling}
Recall that $\Lap Q\ge 0$ on $S$.
The purpose of the present investigation is to study (isolated) points
$p\in \Int S$ at which $\Lap Q(p)=0$. We refer to such points as
\textit{bulk singularities}. Without loss of generality, we can assume that $p=0$ is such a point, and we study the microscopic behaviour of the system $\{\zeta_j\}_1^n$ near $0$.

By the \textit{mesoscopic scale} at $p=0$ we mean the positive number $\meso_n=\meso_n(p)$ having the property
$$n\int_{D(p,\meso_n)}\Lap Q\, dA=1.$$
Intuitively, $\meso_n(p)$ means the expected distance from a particle at $p$ to its closest neighbour. If $p$ is a regular bulk point, then, as is easily seen,
$$\meso_n(p)=1/\sqrt{n\Lap Q(p)}+O(1/n),\quad (n\to\infty),$$
which gives the familiar scaling factor used in papers such as \cite{A,AKM}.

Since the Laplacian $\Lap Q$ vanishes at $0$ and is real-analytic and non-negative in a neighbourhood, there is an integer $k\ge 1$ such that the Taylor expansion of $\Lap Q$ about $0$ takes the form
$\Lap Q(\zeta)=\tilde{P}(\zeta)+O(|\,\zeta\,|^{\,2k-1})$, where
$\tilde{P}(x+iy)=\sum_{j=0}^{2k-2}a_j\,x^{\,j}y^{\,2k-2-j}$
is a positive semi-definite polynomial, homogeneous of degree $2k-2$.

We refer to the number $2k-2=\operatorname{degree}\tilde{P}$ as the \textit{type} of the bulk-singularity at the origin. We shall say that the singularity is \textit{non-degenerate} if
$\tilde{P}$ is positive definite, i.e. if there is a positive constant $c$ such that $\tilde{P}(\zeta)\ge c\,\babs{\,\zeta\,}^{\,2k-2}.$ In the sequel, we tacitly assume that this condition is satisfied.

\smallskip

It will be important to have a good grasp of the size of $r_n=r_n(0)$ as $n\to\infty$. For this, we note that
\begin{align*}1&=n\int_{\babs{\,\zeta\,}<r_n}\Lap Q(\zeta)\, dA(\zeta)\\
&=n\int_0^{r_n}r^{\,2k-1}\, dr\,\frac 1 \pi \int_0^{2\pi}
\tilde{P}(e^{\,i\theta})\, d\theta+O(n\,r_n^{\,2k+1})\\
&=\modul^{-2k}\,n\,r_n^{\,2k}+O(n\,r_n^{\,2k+1})\end{align*}
where $\modul=\modul[Q,0]$ is the positive constant satisfying
\begin{equation}\label{mod}\modul^{-2k}=\frac 1 {2\pi k}\int_0^{2\pi} \tilde{P}(e^{i\theta})\, d\theta.\end{equation}
We will call $\modul$ the \textit{modulus} of the bulk singularity at $0$. We have the following lemma; the simple verification is omitted here.

\begin{lemm} For the mesoscopic scale $r_n$ at $0$ we have $r_n=\modul\, n^{-1/2k}\,(1+O(n^{-1/2k}))$ as $n\to \infty$, where $\modul$ is the modulus \eqref{mod}.
\end{lemm}

\begin{ex} For the \textit{Mittag-Leffler potential} $Q=\babs{\,\zeta\,}^{\,2k}$, the droplet is the disk $\babs{\,\zeta\,}\le k^{-1/2k}$. For $k=1$ we have the well-known \textit{Ginibre potential}. For $k\ge 2$, the Mittag-Leffler potential has a bulk singularity at the origin of type $2k-2$. It is easy to check that the modulus equals $\modul=k^{-1/2k}$.
\end{ex}

Let $p$ be an integer, $1\le p\le n$.
The \textit{$p$-point function} of the point-process $\{\zeta_j\}_1^n$ is the function of $p$ complex variables $\eta_1,\ldots,\eta_p$ defined by
$$\bfR_{n,p}(\eta_1,\ldots,\eta_p)=\lim_{\delta\to 0}\frac {\Prob_n\left(\{\zeta_j\}_1^n\cap D(\eta_\ell,\delta)\ne \emptyset,\,
\ell=1,\ldots,p\right)}{\delta^{2p}}.$$
 The $p$-point function $\bfR_{n,p}$ should really be understood as the density in the measure
$\bfR_{n,p}(\eta_1,\ldots,\eta_p)\, dA(\eta_1)\cdots dA(\eta_p)$. This should be kept in mind when we subject the $\eta_j$ to various transformations.

A well-known algebraic fact ("Dyson's determinant formula'', see e.g. \cite{M} or \cite{ST}, p. 249.) states that the $p$-point function takes the form of a determinant,
$$\bfR_{n,p}(\eta_1,\ldots,\eta_p)=\det\left(\bfK_n(\eta_i,\eta_j)\right)_{i,j=1}^p$$
where $\bfK_n$ is a certain Hermitian function called a \textit{correlation kernel} of the process.
(Cf. Section \ref{Sec2}.)
Of particular importance is the \textit{one-point function}
$\bfR_n=\bfR_{n,1}.$

We now rescale about the origin on the mesoscopic scale $r_n$ about the bulk singularity at $0$. The \textit{rescaled system} $\{z_j\}_1^n$ is taken to be
\begin{equation}\label{scale}z_j=r_n^{-1}\zeta_j,\quad j=1,\ldots,n,\end{equation}
with the law given by the image of the Boltzmann-Gibbs distribution \eqref{bglaw} under the scaling \eqref{scale}.

It follows that the rescaled system $\{z_j\}_1^n$ is determinantal with $p$-point function
\begin{equation}\label{iip}R_{n,p}(z_1,\ldots,z_p)=r_n^{\,2p}\,\bfR_{n,p}(\zeta_1,\ldots,\zeta_p)=
\det(K_{n}(z_i,z_j))_{i,j=1}^p,\end{equation}
where the correlation kernel $K_n$ for the rescaled system is given by
\begin{equation}\label{volk}K_n(z,w)=r_n^{\,2}\,\bfK_n(\zeta,\eta),\qquad (z=r_n^{-1}\zeta,\,w=r_n^{-1}\eta).\end{equation}
In particular, the one-point function of the process $\{z_j\}_1^n$ is $R_n(z)=K_n(z,z).$

Clearly a correlation kernel $K_n(z,w)$ is only determined up to multiplication by a cocycle $c_n(z,w)$.

\subsection{Main structural lemma}
Now suppose that $Q$ has a bulk-singularity of type $2k-2$ at the origin. It will be useful to single out a canonical "dominant part'' of $Q$ near $0$. To this end, let $P(x+iy)$ be the Taylor polynomial of $Q$ of degree $2k$ about the origin. Let $H$ be the holomorphic polynomial
$$H(\zeta)=Q(0)+2\d Q(0)\cdot \zeta+\d^2 Q(0)\cdot\zeta^{\,2}+\cdots+\frac 2 {(2k)!}\d^{2k}Q(0)\cdot \zeta^{\,2k}.$$
We will write
$$Q_0=P-\re H.$$
We then have the basic decomposition
\begin{equation}\label{badem}Q=Q_0+\re H+Q_1\end{equation}
where $Q_1(\zeta)=O(|\,\zeta\,|^{\,2k+1})$ as $\zeta\to 0$.

The following lemma gives the basic structure of limiting kernels at a singular point (not necessarily in the bulk).

\begin{mlem} \label{lema} There exists a sequence $c_n$ of cocycles such that every subsequence of the sequence $c_nK_n$ has a subsequence converging uniformly on compact subsets to some Hermitian function $K$. Every limit point $K$ has the structure
\begin{equation}\label{Mod}K(z,w)=L(z,w)e^{-Q_0(\modul z)/2-Q_0(\modul w)/2}\end{equation} where $L$ is an Hermitian-entire function.
\end{mlem}

Following \cite{AKM}, we refer to a limit point $K$ in Lemma \ref{lema} as a \textit{limiting kernel} whereas $L$ is a \textit{limiting holomorphic kernel}. We also speak of the \textit{limiting 1-point function}
\begin{equation}\label{polar}R(z)=K(z,z)=L(z,z)e^{-Q_0(\modul z)}.\end{equation}
Note that $R$ determines $K$ and $L$ by polarization.

\begin{rem} Each limiting one-point function gives rise to a unique \textit{limiting point field} (or "infinite particle system'') $\{z_j\}_1^\infty$ with intensity functions
$$R_k(z_1,\ldots,z_p)=\det(K(z_i,z_j))_{i,j=1}^p.$$ (This follows from Lenard's theory, see \cite{S} or \cite{AKM}.) It is possible that a limiting point field is
\textit{trivial} in the sense that $K=0$.
\end{rem}

\subsection{Universality results}
We will prove universality for two kinds of bulk singularities. Referring to the canonical decomposition $Q=Q_0+\re H+Q_1$ with
$Q_0$ of degree $2k$, we say a singularity at $0$ is:
\begin{enumerate}[label=(\roman*)]
\item \textit{homogeneous} if $Q_1=0$ and $H(z)=c\,z^{\,2k}$ for some constant $c$,
\item \textit{dominant radial} if $Q_0$ is radially symmetric, i.e. $Q_0(z)=Q_0(|\,z\,|)$.
\end{enumerate}
We remark that a homogeneous singularity is necessarily located in the bulk of the droplet; for other types of singularities this must be postulated.

In the following we denote by $L_0$ the Bergman kernel of the space of entire functions $L^2_a(\mu_0)$
associated with the measure
\begin{equation}\label{fock0}d\mu_0(z)= e^{-Q_0(\modul z)}\, dA(z).\end{equation}

\begin{mth} \label{mth2}
If there is a homogeneous singularity at $0$ we have $L=L_0$ for each limiting holomorphic kernel $L$.
\end{mth}

The next result concerns limiting holomorphic kernel $L(z,w)$ which are \textit{rotationally symmetric} in the sense that $L(z,w)=L(ze^{it},we^{it})$ for all real $t$. Equivalently, $L$ is rotationally symmetric if there is an entire function $E$ such that $L(z,w)=E(z\bar{w})$. (We leave the simple verification of this to the reader.)

\begin{mth} \label{mth3} If a bulk singularity at $0$ is dominant radial, then $L=L_0$ for each rotationally symmetric limiting kernel.
\end{mth}

The result was conjectured in \cite{AKM}, Section 7.3.

We do not know whether or not each limiting kernel at a dominant radial bulk singularity is rotationally symmetric. This question seems to be related to the problem of deciding the \textit{translation invariance} of limiting kernels at regular boundary points. See \cite{AKM} for several comments about this, notably the interpretation in terms of a twisted convolution equation in Section 7.1. It is natural to conjecture that the kernel in Theorem \ref{mth2} be equal to the limiting kernel in general, regardless of the nature of a (non-degenerate) bulk singularity.

\begin{figure}[ht]
\begin{center}

\includegraphics[width=.4\textwidth]{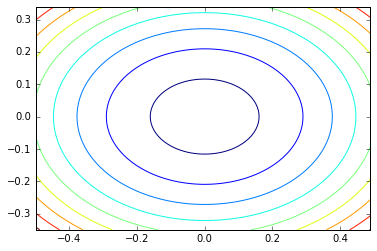}
\hspace{.1\textwidth}
\includegraphics[width=.38\textwidth]{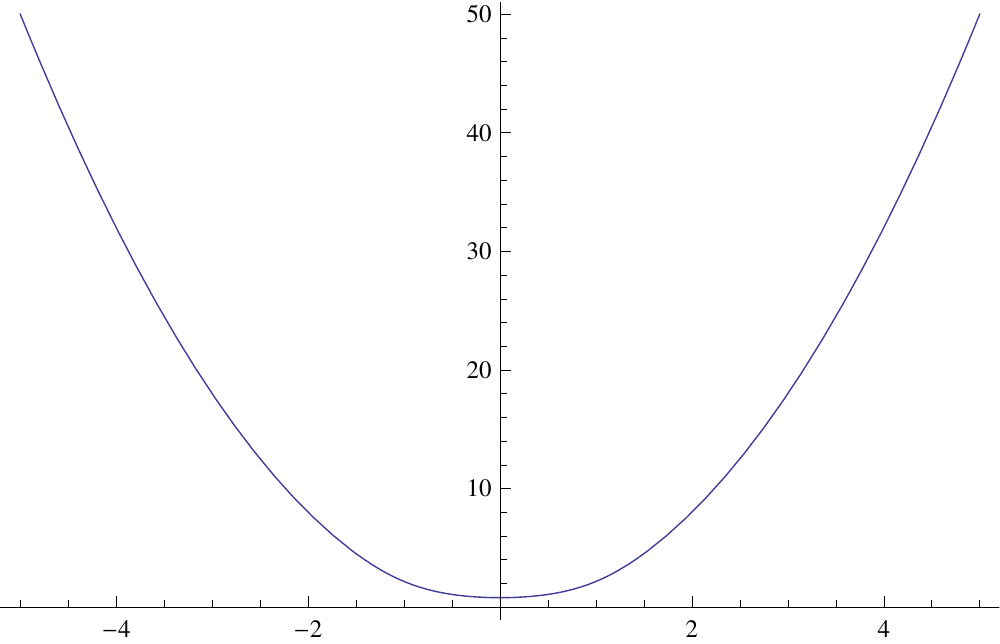}
\end{center}

\caption{Some level curves of $R_0(z)=L_0(z,z)e^{-Q_0(\modul z)}$ for $Q_0(z)=|\,z\,|^{\,4}-|\,z\,|^{\,2}\,\re(\,z^{\,2}\,)/2$
and the graph of $R_0(x)=M_2(x^{\,2})\,e^{-Q_0(\modul x)}$ for the Mittag-Leffler potential $Q_0(z)=\babs{\,z\,}^{\,4}$.}
\end{figure}

\begin{rem}
Note that, as a consequence of the reproducing property of the kernel $L_0$, we have in the situation of the above theorems the \textit{mass-one equation} for a limiting kernel $K$, $\int_\C\babs{\,K(z,w)\,}^{\,2}\, dA(w)=R(z)$.
\end{rem}

\begin{ex} For the Mittag-Leffler potential $Q=|\,\zeta\,|^{\,2k}$ it is possible to calculate the limiting kernel $L$ explicitly, using orthogonal polynomials (see \cite{AKM}, Section 7.3). The result is that
\begin{equation}\label{mll}L(z,w)=M_k(z\bar{w}),\end{equation}
where $$M_k(z)=\modul^{\,2}\,k\sum_0^\infty \frac {(\modul^{\,2} z)^{\,j}}{\Gamma\left(\frac {1+j}k\right)}.$$
The function $M_k$ can be expressed as
$M_k(z)=\modul^{\,2}\,k\, E_{1/k,1/k}(\modul^{\,2} z)$
 where $E_{a,b}$ is the Mittag-Leffler function (see \cite{GKMR})
\begin{equation}\label{twopm}E_{a,b}(z)=\sum_0^\infty\frac {z^{\,j}}{\Gamma(aj+b)}.\end{equation}

Using Theorem \ref{mth2} we can now see that the kernel in \eqref{mll} is universal for potentials of the form
$Q=|\,\zeta\,|^{\,2k}+\re\left(c\,\zeta^{\,2k}\right)$. (We must insist that $|\,c\,|<1$ to insure that the growth assumption of $Q$ at infinity is satisfied, see \eqref{eq:grow}.) By Theorem \ref{mth3} the universality holds also for all rotationally symmetric limiting kernels $L(z,w)=E(z\bar{w})$ for more general potentials of the form
$Q(\zeta)=\babs{\,\zeta\,}^{\,2k}+ \re H(\zeta)+Q_1(\zeta)$.
\end{ex}

\begin{rem} For $k=1$ (i.e. when $0$ is a "regular'' bulk point) the space $L^2_a(\mu_0)$ becomes the standard Fock space, normed by
$\|\,f\,\|^{\,2}=\int_\C|\,f(z)\,|^{\,2}e^{-\,|\,z\,|^{\,2}}\, dA(z)$. In this case we have $R=1$ for the limiting 1-point function, by the well-known Ginibre$(\infty)$-limit. (See e.g. \cite{AKM}.)
\end{rem}

\subsection{Further results} In the following, we consider a potential with canonical decomposition $Q=Q_0+\re H+Q_1$. Following \cite{AKM}, we shall prove auxiliary results which fall in three categories.

\subsubsection*{Ward's equation}
Let $R(z)=K(z,z)$ be a limiting kernel in Lemma \ref{lema}. At a point $z$ where $R>0$, we put
\begin{align}B(z,w)&=\frac {\babs{\,K(z,w)\,}^{\,2}}{K(z,z)}=\frac {\babs{\,L(z,w)\,}^{\,2}}{L(z,z)}e^{-Q_0(\modul w)},\\
\quad C(z)&=\int_\C \frac {B(z,w)}{z-w}\, dA(w).
\end{align}
We call $B(z,w)$ a \textit{limiting Berezin kernel} rooted at $z$; $C(z)$ is its \textit{Cauchy transform}.

\begin{mth} \label{01ward} Let $R$ be a limiting $1$-point function.
\begin{enumerate}[label=(\roman*)]
\item Zero-one law: Either $R=0$ identically, or else $R>0$ everywhere.
\item Ward's equation: If $R$ is non-trivial, we have that
\begin{equation}\label{cowa}\dbar C(z)=R(z)-\Lap_z \left[Q_0(\modul z)\right]-\Lap_z\log R(z).\end{equation}
\end{enumerate}
\end{mth}

As $n\to\infty$ it may well happen that $R_n\to 0$ locally uniformly (if the singularity at $0$ is in the exterior of the droplet).

\begin{figure}[ht]
\begin{center}

\includegraphics[width=.4\textwidth]{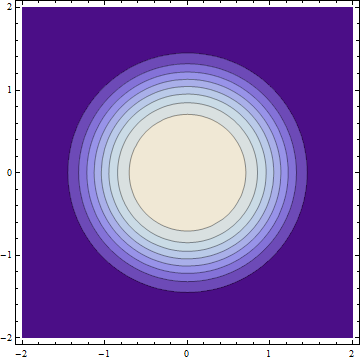}
\hspace{.1\textwidth}
\includegraphics[width=.4\textwidth]{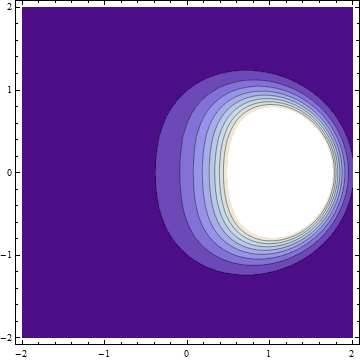}
\end{center}

\caption{Some level curves of the Berezin kernel $B(z,w)$ rooted at $z=0$ and $z=1$
pertaining to $Q_0(z)=\babs{\,z\,}^{\,4}$.}

\end{figure}

\subsubsection*{Apriori estimates} To rule out the possibility of trivial limiting kernels, we shall use the following result.

\begin{mth} \label{apth} Let $R$ be any limiting kernel, and let $R_0(z)=L_0(z,z)e^{-Q_0(\modul z)}$ where $L_0$ is the Bergman kernel
of the space $L^2_a(\mu_0)$. Then
\begin{enumerate}[label=(\roman*)]
\item \label{r0p} $R_0(z)=\Lap_z [Q_0(\modul z)]\cdot (1+O(z^{1-k}))$, as $z\to\infty$,
\item \label{r1p} $R(z)\,\,\,=\Lap_z [Q_0(\modul z)]\cdot (1+O(z^{1-k}))$, as $z\to\infty$.
\end{enumerate}
\end{mth}

Part \ref{r0p} depends on an estimate of the Bergman kernel for the space $L^2_a(\mu_0)$. Related estimates
valid when $Q_0$ is a function satisfying uniform estimates of the type $0<c\le\Lap Q_0\le C$ are found in Lindholm's paper \cite{L}.

In our situation, the function $\Lap Q_0$ takes on all values between $0$ and $+\infty$, which means that the results from \cite{L} are not directly applicable. It has turned out convenient to include an elementary discussion for the case at hand, following the method of "approximate Bergman projections'' in the spirit of \cite{AKM}, Section 5. This has the advantage that proof of part \ref{r1p} follows after relatively simple modifications.

\begin{rem} Part \ref{r0p} of Theorem \ref{apth} seems to be of some relevance for the investigation of density conditions for sampling and interpolation in Fock-type spaces $L^2_a(\mu_0)$; see the recent paper \cite{FGHKR}, Remark 5.6. (A very general result of this sort was obtained by different methods in the paper \cite{MMOC}, where the hypothesis on the "weight" $Q_0$ is merely that the Laplacian $\Lap Q_0$ be a doubling measure.)
\end{rem}

\begin{rem}
In the case $Q=|\,z\,|^{\,2\lambda}$, the asymptotic formula in Theorem \ref{apth} \ref{r0p} has an alternative proof by more classical methods, using an asymptotic expansion for the function $M_\lambda(z)$ as $z\to\infty$ (\cite{GKMR}, Section 4.7). The formula \ref{r0p} can be recognized as giving the leading term in that expansion.
\end{rem}

\subsubsection*{Positivity} Recall that a Hermitian function $K$ is called a \textit{positive matrix} if
$$\sum_{i,j=1}^N\alpha_i\bar{\alpha}_j K(z_i,z_j)\ge 0$$ for all points $z_j\in\C$ and all complex scalars $\alpha_j$. It is clear that each limiting (holomorphic) kernel is a positive matrix.

\begin{mth} \label{pthm} Let $L$ be a limiting holomorphic kernel. Then $L$ is the Bergman kernel for a Hilbert space $\calH_*$ of entire functions which sits contractively in $L^2_a(\mu_0)$. Moreover, $L_0-L$ is a positive matrix.
\end{mth}

Here $L_0$ is the Bergman kernel of $L^2_a(\mu_0)$. It may well happen that the space $\calH_*$ degenerates to $\{0\}$. This is the case when the singularity at $0$ is located in the exterior of the droplet.

\subsubsection*{Comments}

An interesting generalization of our situation is obtained by allowing for a suitably scaled logarithmic singularity at a (regular
or singular) bulk point. More precisely, if $\tilde{Q}$ is a smooth in a neighbourhood of $0$, we consider a
potential of the form $Q(\zeta)=\tilde{Q}(\zeta)+2(c/n)\log\babs{\,\zeta\,}$ where $c<1$ is a constant.
Rescaling by $z=r_n^{-1}\zeta$ where $c+n\int_{D(0,r_n)}\Lap \tilde{Q}\, dA=1$, we find $r_n\sim (1-c)^{1/2k}\modul n^{-1/2k}$ as $n\to\infty$, where $2k-2$ is the type of $\tilde{Q}$ and $\modul=\modul[\tilde{Q},0]$.
It is hence natural to define the dominant part by $Q_0(z):=c'\,\tilde{Q}_0(z)+2c\log|\,z\,|$, where
$\tilde{Q}_0$ is the dominant part of $\tilde{Q}$ and $c'$ a suitable constant depending on $c$.
In particular, if $Q(\zeta)=c_1|\,\zeta\,|^{\,2\lambda}+(c_2/n)\log\babs{\,\zeta\,}$ with suitable $c_1,c_2>0$,
the dominant part becomes of the type
\begin{equation}\label{pppart}Q_0(z)=r^{\,2\lambda}+2\left(1-\frac \lambda\mu\right)\log r,\qquad r=|\,z\,|,\end{equation}
for suitable constants $\lambda$ and $\mu$.
The potential \eqref{pppart}
was introduced in the paper \cite{AK}, where all rotationally symmetric solutions to the corresponding Ward equation \eqref{cowa} were found. Recently, certain potentials of this form were studied in a context of Riemann surfaces, in a scaling limit about certain types of singular points (conical singularities and branch points) see \cite{LCCW}.
We will return to this issue in a forthcoming paper \cite{AKS}.

As in \cite{AKM}, Section 7.7, we note that it is possible to introduce an "inverse temperature'' $\beta$ into the setting; the case at hand then corresponds to $\beta=1$. For general $\beta$, the rescaled process $\{z_j\}_1^n$ is no longer determinantal, but the rescaled intensity functions $R_{n,p}^\beta$ make perfect sense. As $n\to\infty$, we  \textit{formally} obtain a "Ward's equation at a bulk singularity'' of the form
\begin{equation}\label{bewa}\dbar C^\beta(z)=R^\beta(z)-\Lap_z\left[Q_0(\modul z)\right]-\frac 1 \beta\Lap_z\log R^\beta(z).\end{equation}
Here $C^\beta(z)$ should be understood as the Cauchy transform of the $\beta$-Berezin kernel $B^\beta(z,w)=(R_1^\beta(z)R_1^\beta(w)-R_2^\beta(z,w))/R_1^\beta(z)$.
The objects in \eqref{bewa} are so far understood mostly on a physical level. We now give a few remarks in this spirit.

First, if $0$ is a regular bulk-point, i.e. if $\Lap Q(0)>0$, then it is believed that $R^\beta=1$ identically, i.e., the right hand side in \eqref{bewa} should vanish. The equation \eqref{bewa} then reflects the fact that the Berezin kernel $B^\beta(z,w)=b^\beta(r)$ depends only on the distance $r=|\,z-w\,|$. When $\beta=1$ one has the well-known identity $b^1(r)=e^{-\,r^{\,2}}$. For other $\beta$ we do not know of an explicit expression, but
it was shown by Jancovici in \cite{J} that
\begin{equation*}b^\beta(r)=b^1(r)+(\beta-1)f(r)+O(\,(\beta-1)^{\,2}\,),\qquad (\beta\to 1)\end{equation*}
where $f$ is a certain explicit function.
 In the bulk-singular case, the kernel $B^\beta(z,w)$ will not just depend on $|\,z-w\,|$, but still it seems natural to expect that we have an expansion of the form
\begin{equation}\label{expii}B^\beta(z,w)=b^1(z,w)+(\beta-1)f(z,w)+O(\,(\beta-1)^{\,2}\,),\qquad (\beta\to 1)\end{equation}
where $b^1(z,w)=\babs{\,L_0(z,w)\,}^{\,2}e^{-Q_0(\modul w)}/L_0(z,z)$, $L_0$ being the Bergman kernel of the space $L^2_a(\mu_0)$. A natural problem, which will not be taken up here, is to determine the function $f(z,w)$ in \eqref{expii}. (A similar investigation at regular boundary points was made recently in the paper \cite{CFTW}.)

For boundary points, the term "singular'' has a different meaning than for bulk points. Indeed, the singular points $p$ (cusps or double points) studied in the paper \cite{AKMW} all satisfy $\Lap Q(p)>0$. An example of a situation
at which $\Lap Q=0$ at a boundary point (at $0$)
is provided by the potential $Q=|\,\zeta\,|^{\,4}-\sqrt{2}\re(\,\zeta^{\,2}\,)$. (The boundary of $S$ is here a "figure 8'' with $0$ at the point of self-intersection, see \cite{BGM}.) A natural question is whether it is possible to define non-trivial scaling limits at (or near) this kind of singular points, in the spirit of \cite{AKMW}.

There is a parallel theory for scaling limits for Hermitian random matrix ensembles. In this situation, the droplet is a union of compact intervals. It is well known that the sine-kernel appears in the scaling limit
about a "regular bulk point'', i.e. an interior point where the density of the equilibrium measure is strictly positive. In a generic case, all points are regular, see \cite{KM}. Special bulk points where the equilibrium density vanishes may be called "singular''; at such points other types of universality classes appear, see \cite{BE,C,CKV,PS}.

Finally, we wish to mention that the investigations in this paper were partly motivated by applications to the distribution of Fekete points close to a bulk singularity (see \cite{A}). This issue will be taken up in a later publication.

\subsection{Plan of the paper} In Section \ref{Sec2} we prove the general structure formula for limiting kernels (Lemma \ref{lema}). We also prove the positivity theorem (Theorem \ref{pthm}).

In Section \ref{Sec3} we prove Ward's equation and the zero-one law (Theorem \ref{01ward}).

In Section \ref{Sec4} we prove the universality results (theorems \ref{mth2} and \ref{mth3}). Our proof of Theorem \ref{mth3} depends on the apriori estimate from Theorem \ref{apth}, part \ref{r1p}.

In the last two sections, we prove the asymptotics for the functions $R_0$ and $R$ in Theorem \ref{apth}. For $R_0$, (part \ref{r0p}) see Section \ref{Sec4.5}; for $R$, (part \ref{r1p}) see Section \ref{Sec5}.

\subsection{Convention} Multiplying the potential $Q$ by a suitable constant, we can in the following assume that the modulus $\modul=1$. In fact, the slightly more general assumption that $\modul=1+O(n^{-1/2k})$ as $n\to\infty$ will do equally well. This means the mesoscopic scale about $0$ can be taken as $r_n=n^{-1/2k}$, where $2k-2$ is the type of the singularity. In the sequel, this will be assumed \textit{throughout}.

\section{Structure of limiting kernels} \label{Sec2}

In this section, we prove Lemma \ref{lema} on the general structure of limiting kernels and the positivity theorem
\ref{pthm}. We shall actually prove a little more: a limiting holomorphic kernel can be written as a subsequential limit of kernels for certain specific Hilbert spaces of entire functions. In later sections, we will use this additional information for our analysis of homogeneous bulk singularities.

\subsection{Spaces of weighted polynomials}
It is well-known that we can take for correlation kernel for the process $\{\zeta_j\}_1^n$ the reproducing kernel for a suitable space of weighted polynomials. Here the "weight'' can either be incorporated into the polynomials themselves, or into the norm of the polynomials. We will use both these possibilities. In the following we shall use the symbol "$\Poly(n)$'' for the linear space of holomorphic polynomials of degree at most $n-1$ (without any topology). We write $\mu_n$ for the measure $d\mu_n=e^{-nQ}\, dA$.

We let $\calP_n$ denote the space $\Poly(n)$ regarded as a subspace of $L^2(\mu_n)$. The symbol $\calW_n$ will denote the set of weighted polynomials $f=p\,e^{-nQ/2}$, ($p\in\Poly(n)$) regarded as a subspace of $L^2=L^2(dA)$. We write
$\bfk_n$ and $\bfK_n$ for the reproducing kernels of $\calP_n$ and $\calW_n$ respectively, and we note that
$$\bfK_n(\zeta,\eta)=\bfk_n(\zeta,\eta)\,e^{-nQ(\zeta)/2-nQ(\eta)/2}.$$

Now suppose that $Q$ has a bulk singularity at the origin, of type $2k-2$ and rescale at the mesoscopic scale by
$$k_n(z,w)=r_n^{\,2}\,\bfk_n(\zeta,\eta),\quad K_n(z,w)=r_n^{\,2}\,\bfK_n(\zeta,\eta),\qquad (z=r_n^{-1}\zeta,\,\,w=r_n^{-1}\eta).$$

\subsection{Limiting holomorphic kernels} Suppose that there is a bulk singularity of type $2k-2$ at the origin.
Consider the canonical decomposition $Q= Q_0 +\re H+ Q_1$ and write
$h=\re H$. Thus $h$ is of degree at most $2k$, $Q_0$ is a positive definite homogeneous polynomial of degree $2k$, and
$Q_1(\zeta)=O(\babs{\,\zeta\,}^{\,2k+1})$ as $\zeta \to 0$.
\begin{lem} \label{lem:bound}
For each compact subset $\compset$ of $\C$, there is a constant $C=C(\compset)$ such that
$K_n(z,z) \leq C$ for $z \in \compset$.
\end{lem}
\begin{proof}
Let $\tilde{\calW}_n$ denote the space of all "rescaled'' weighted polynomials $p \cdot e^{-\tilde{Q}_n/2}$ where $p\in \Poly(n)$ and $\tilde{Q}_n(z)= nQ(r_n z)$. Regarding $\tilde{\calW}_n$ as a subspace of $L^2$, we recognize that $K_n$ is the reproducing kernel of $\tilde{\calW}_n$. Hence
\begin{align}\label{ker:sup}
K_n(z,z)=\sup\{\,|\,f(z)\,|^{\,2} \, ; \, f \in \tilde{\calW}_n,\, \|\, f\, \|\leq 1\, \}.
\end{align}
Fix a number $\delta>0$ and let $\compset_{\delta}= \{\, z \in \C \, ;\, \dist(z,\compset)\leq \delta\, \}$. We also pick a number $\alpha > \sup \{\, \Lap Q_0(z)\, \ ; \ z \in \compset_{\delta}\, \}$. Now let $u$ be an analytic function in a neighbourhood of $\compset_{\delta}$ and consider the function $g_n(z)= u(z)\, e^{-\,\tilde{Q}_n (z)/2+\,\alpha\, \babs{\,z\,}^{\,2}/2}$. Note that
\begin{equation*}\Lap \tilde{Q}_n(z) = n\, r_n^{\,2} (\Lap Q_0(r_n z) + \Lap Q_1(r_n z))=n\, r_n^{\,2k}\,
\Lap Q_0(z)+O(n \,r_n^{\,2k+1}),\quad (n\,r_n^{\,2k}=1).
\end{equation*}
 Hence
  $\Lap \log |\,g_n(z)\,|^{\,2} \geq  - \Lap \tilde{Q}_n(z) + \alpha>0$ for all sufficiently large $n$ and all $z \in \compset_{\delta}$. Thus $\babs{\,g_n\,}^{\,2}$ is subharmonic in $\compset_\delta$, so for $z \in \compset$
\begin{align*}
	\babs{\,g_n(z)\,}^{\,2} &\leq \delta^{-2} \int_{D(z,\delta)} \babs{\,g_n (w)\,}^{\,2} dA(w)\\
	&= \delta^{-2} e^{\,\alpha\, \left(|\,z\,|+\delta\right)^{\,2}}
	\int_{D(z,\delta)} \babs{\,u(w)\,}^{\,2} e^{- \tilde{Q}_n(w)} dA(w).
\end{align*}
We obtain
\begin{align} \label{sub:ineq}
	\babs{\,u(z)\,}^{\,2} e^{-\tilde{Q}_n(z)} \leq \delta^{-2} e^{\,\alpha\,\left(\,2\,M_{\compset}\,\delta\,+\, \delta^{\,2}\,\right)} \int_{D(z,\delta)} |\,u\,|^{\,2}e^{-\tilde{Q}_n} dA,
\end{align}
where $M_{\compset}=\sup_{z\in \compset}{|\,z\,|}$. By (\ref{ker:sup}) and (\ref{sub:ineq}), $K_n(z,z)$ is bounded for $z\in\compset$.
\end{proof}

We now use the holomorphic polynomial $H$ in the decomposition $Q=Q_0+\re H+Q_1$ to
define a Hermitian-entire function ("rescaled holomorphic kernel'') by
\begin{equation}\label{rhk}L_n(z,w)=r_n^{\,2}\,\bfk_n(\zeta,\eta)\,e^{-n(H(\zeta)+\bar{H}(\eta))/2},\qquad z=r_n^{-1}\zeta,\, w=r_n^{-1}\eta.\end{equation}
Let us write
\begin{equation}H_n(z)=n\,H(r_nz),\qquad Q_{1,n}(z)=n\,Q_1(r_nz),\end{equation}
so that $nQ(r_nz)=Q_0(z)+\re H_n(z)+Q_{1,n}(z)$ and
$$L_n(z,w)=k_n(z,w)\,e^{-H_n(z)/2-\bar{H}_n(w)/2}.$$

Define a Hilbert space of entire functions by
\begin{equation}\label{heff}\calH_n=\{f=q\cdot e^{-H_n/2};\, q\in\Poly(n)\}\end{equation}
equipped with the norm of $L^2(\tilde{\mu}_n)$ where
\begin{equation}\label{apmeas}d\tilde{\mu}_n(z)=e^{-Q_0(z)-Q_{1,n}(z)}\, dA(z).\end{equation}
Observe that $Q_{1,n}=O(r_n)$ as $n\to\infty$ where the $O$-constant is uniform on each given compact subset of $\C$. In particular $\tilde{\mu}_n\to\mu_0$ vaguely where
$d\mu_0=e^{-Q_0}\, dA.$

The following result implies Lemma \ref{lema}; it also generalizes Lemma 4.9 in \cite{AKM}.

\begin{lem} \label{lem:compact} Each subsequence of the kernels $L_n$ has a further subsequence converging
locally uniformly to a Hermitian-entire limit $L$. Furthermore, $L_n$ is the reproducing kernel of the space $\calH_n$, and $L$ satisfies the "mass-one inequality'',
\begin{align} \label{mass:ineq}
	\int \babs{\,L(z,w)\,}^{\,2}\, d\mu_0(w) \leq L(z,z).
\end{align} Finally,
there exists a sequence of cocycles $c_n$ such that each subsequence of $c_nK_n$ converges locally uniformly to a Hermitian function $K$ of the type $K(z,w)=L(z,w)e^{-Q_0(z)/2-Q_0(w)/2}$,
\end{lem}

\begin{proof}
Define a function $E_n(z,w)$ by
\begin{align*}E_n(z,w)&= e^{\,n(H(\zeta)/2+\bar{H}(\eta)/2-Q(\zeta)/2-Q(\eta)/2)}\\
&=e^{-Q_0(z)/2-Q_0(w)/2-Q_{1,n}(z)/2-Q_{1,n}(w)/2
+i\im( H_n(z)-H_n(w))/2}.
\end{align*}
Note that $K_n= L_n E_n$ where $L_n$ is the Hermitian-entire kernel \eqref{rhk}.
Now, if $h=\re H$ then
\begin{align*}
\im(H_n(z)-H_n(w))/2= \sum_{j=1}^{2k} n\,r_n^{\,j}
\im \left(  \frac{\d^j h(0)}{j!} (z^{\,j}-w^{\,j})  \right) .
\end{align*}
We have shown that
$$E_n(z,w)= c_n(z,w)\,e^{-Q_0(z)/2-Q_0(w)/2}\,(1+o(1)),\quad (n\to\infty)$$ where $o(1)\to 0$ locally uniformly on $\C^2$ and $c_n$ is a cocycle:
$$c_n(z,w)=\prod_{j=1}^{2k}\exp\left[i\,n\,r_n^{\,j}\im\left(  \frac{\d^j h(0)}{j!} (z^{\,j}-w^{\,j})  \right)\right].$$

On the other hand, for each compact subset $\compset$ of $\C^2$ there is a constant $C$ such that
$$\babs{\,L_n(z,w)\,}^{\,2} = \babs{\, \frac{K_n(z,w)}{ E_n(z,w)}\,}^{\,2}
\leq C  K_n(z,z)\, K_n(w,w)\, e^{\,Q_0(z)+Q_0(w)}$$
for sufficiently large $n$. By Lemma \ref{lem:bound}, the functions $L_n$ have a uniform bound on $\compset$. We have shown that $\{L_n\}$ is a normal family. We can hence extract a subsequence $\{ L_{n_\ell} \}$, converging locally uniformly to a Hermitian-entire function $L(z,w)$.

Choosing cocycles $c_n$ such that $c_{n} E_n \to e^{-Q_0(z)/2-Q_0(w)/2}$ uniformly on compact subsets as $n\to\infty$, we now obtain that
$$c_{n_\ell}K_{n_\ell}= c_{n_\ell}E_{n_\ell}L_{n_\ell} \to e^{-Q_0(z)/2-Q_0(w)/2} L(z,w)=K(z,w).$$

The reproducing property $\int |\,K_{n_\ell}(z,w)\,|^{\,2} dA(w)= K_{n_\ell}(z,z)$ means that
\begin{align*}
\int \babs{\,E_{n_\ell}(z,w) L_{n_\ell}(z,w)\,}^{\,2} dA(w) =  E_{n_\ell}(z,z) L_{n_\ell}(z,z).
\end{align*}
Letting $\ell \to \infty$, we obtain the mass-one inequality (\ref{mass:ineq}) by Fatou's lemma.

There remains to prove that $L_n$ is the reproducing kernel for the space $\calH_n$. For this, we write $L_{n,w}(z)=L_n(z,w)$ and note that
for an element $f=q\cdot e^{-H_n/2}$ of $\calH_n$ we have
\begin{align*}\langle f,L_{n,w}\rangle_{L^2(\tilde{\mu}_n)}&=\int_\C q(z)\,e^{-H_n(z)/2}\,\bar{L}_n(z,w)\,e^{-Q_0(z)-Q_{1,n}(z)}\,dA(z)\\
&=e^{-H_n(w)/2}\int_\C q(z)\,\bar{k}_n(z,w)\,e^{-nQ(r_nz)}\, dA(z).
\end{align*}
Noting that $k_n$ is the reproducing kernel for the space $\tilde{\calP}_n$ of polynomials of degree at most $n-1$ normed by $\|\,p\,\|^{\,2}=\int_\C\babs{\,p(z)\,}^{\,2}e^{-nQ(r_nz)}\, dA(z)$, we now see that
$$\langle f,L_{n,w}\rangle_{L^2(\tilde{\mu}_n)}=e^{-H_n(w)/2}q(w)=f(w).$$
The proof of the lemma is complete. \end{proof}

\subsection{The positivity theorem} Let $\mu_0$ be the measure $d\mu_0=e^{-Q_0}\, dA$ and define $L_0(z,w)$ to be the Bergman kernel for the Bergman space $L^2_a(\mu_0)$. Let $L=\lim L_{n_\ell}$ be a limiting holomorphic kernel at $0$.

Recall that the kernel $L_n$ is the reproducing kernel for a certain subspace $\calH_n$ of $L^2_a(\tilde{\mu}_n)$, where $\tilde{\mu}_n\to\mu_0$ in the sense that the densities converge uniformly on compact sets, as $n\to\infty$.
See Lemma \ref{lem:compact}.

For $L=\lim L_{n_\ell}$, the assignment $\langle L_z,L_w\rangle_*=L(w,z)$ defines a positive semi-definite inner product on the linear span $\calM$ of the $L_z$'s. In fact, the inner product is either trivial ($L(z,z)=0$ for all $z$), or else it is positive definite: this holds by the zero-one law in Theorem \ref{01ward}, which will be proved in the next section.

By Fatou's lemma, we now see that, for all choices of points $z_j$ and scalars $\alpha_j$,
\begin{align*}\left\|\,\sum_{j=1}^N\alpha_jL_{z_j}\,\right\|_{L^2(\mu_0)}^{\,2}&\le\liminf_{\ell\to\infty}
\sum_{i,j=1}^N\alpha_i\bar{\alpha_j}\int_\C L_{n_\ell}(w,z_i)\bar{L}_{n_\ell}(w,z_j)\, d\tilde{\mu}_{n_\ell}(w)\\
&=\liminf_{\ell\to\infty}\sum_{i,j=1}^N\alpha_i\bar{\alpha}_jL_{n_\ell}(z_i,z_j)=\sum_{i,j=1}^N\alpha_i\bar{\alpha}_j
L(z_i,z_j)\\
&=\left\|\,\sum_{j=1}^N\alpha_jL_{z_j}\,\right\|_*^{\,2}.
\end{align*}
This shows that $\calM$ is contained in $L^2(\mu_0)$ and that the inclusion $I:\calM\to L^2(\mu_0)$ is a contraction.
Hence the completion $\calH_*$ of $\calM$ can be regarded as a contractively embedded subspace of $L_a^2(\mu_0)$.

Since the space $L^2_a(\mu_0)$ has reproducing kernel $L_0(z,w)$, it follows from a theorem of Aronszajn (\cite{Ar}, p. 355) that the difference $L_0-L$ is a positive matrix. The proof of Theorem \ref{pthm} is complete. q.e.d.

\section{Ward's equation and the zero-one law} \label{Sec3}

\subsection{Ward's equation}
Given a limiting kernel $K$ in Lemma \ref{lem:compact}, we recall the definitions
$$R(z)=K(z,z),\quad B(z,w)=\frac {\babs{\,K(z,w)\,}^{\,2}}{K(z,z)},\quad C(z)=\int_\C\frac {B(z,w)}{z-w}\, dA(w).$$

The goal of this section is to prove
Theorem \ref{01ward}, which we here restate in the following form (the case $\modul=1$).

\begin{lem} \label{bsw} If $R$ does not vanish identically, then $R>0$ everywhere and we have
$$\dbar C(z)=R(z)-\Lap Q_0(z)-\Lap\log R(z).$$
\end{lem}

For the proof of Lemma \ref{bsw}, we recall the setting of Ward's identity from \cite{AKM}.

For a test function $\psi \in C_0^{\infty}(\C)$, we define a function $W_{n}^{+}[\psi]$ of $n$ variables by
\begin{align*}
W_{n}^{+}[\psi] = I_n[\psi]-II_n[\psi]+III_n[\psi],
\end{align*}
where
\begin{align*}
&I_n[\psi](\zeta)= \frac{1}{2} \sum_{j\ne k}^{n} \frac{\psi(\zeta_j)-\psi(\zeta_k)}{\zeta_j- \zeta_k}, \quad II_n[\psi](\zeta)=n\sum_{j=1}^{n} \partial Q(\zeta_j) \cdot \psi(\zeta_j), \quad \mathrm{and} \\
& III_n[\psi](\zeta) = \sum_{j=1}^{n}\partial \psi (\zeta_j) \quad \mathrm{for} \quad \zeta=(\zeta_1, \cdots, \zeta_n) \in \C^n.
\end{align*}
We now regard $\zeta$ as picked randomly with respect to the Boltzmann-Gibbs distribution \eqref{bglaw}.  $W_n^+[\psi]$ is then a random variable;
the Ward identity proved in \cite{AKM}, Section 4.1 states that its expectation vanishes:
\begin{equation}\label{W:identity}\E_n W_n^{+}[\psi]= 0.\end{equation}

We shall now rescale in Ward's identity about $0$ at the mesoscopic scale $r_n=n^{-1/2k}$, given that
the basic decomposition
$Q=Q_0+\re H+Q_1$ in \eqref{badem} holds. (We do not need to assume that $0$ is in the bulk at this stage.)

To facilitate for the calculations, it is convenient to recall a simple
algebraic fact (see e.g. \cite{M}): if $f$ is a function of $p$ complex variables, and if $f(\zeta_1,\ldots,\zeta_p)$ is regarded as a random variable on the sample space $\{\zeta_j\}_1^n$ with respect to the Boltzmann-Gibbs law, then the expectation is
\begin{equation}\label{exp:var}\E_n\left[f(\zeta_1,\ldots,\zeta_p)\right]=\frac {(n-p)!}{n!}\int_{\C^p}f\cdot\bfR_{n,p}\, dV_p\end{equation}
where $dV_p(\zeta_1,\ldots,\zeta_p)=dA(\zeta_1)\cdots dA(\zeta_p)$.

We rescale about $0$ via $z=r_n^{-1}\zeta$, $w=r_n^{-1}\eta$, recalling that the $p$-point functions transform as densities. We remind that $R_{n,p}(z)=r_n^{\,2p}\,\bfR_{n,p}(\zeta)$ denotes the rescaled $p$-point function and use the abbreviation $R_n=R_{n,1}$ for the one-point function. We also write
\begin{align}\label{bern}B_n(z,w)&=\frac{R_n(z)R_n(w)-R_{n,2}(z,w)}{R_n(z)}=\frac{\babs{\,K_n(z,w)\,}^{\,2}}{R_n(z)},\\
\label{cauchyn}C_n(z)&=\int\frac {B_n(z,w)}{z-w}\, dA(w).
\end{align}

\begin{lem} We have that
\begin{align*}
	\bar{\partial} C_n(z) = R_n(z) - \Lap Q_0(z) - \Lap \log R_n(z) +o(1),
\end{align*}
where $o(1) \to 0$ uniformly on compact subsets of $\C$ as $n\to\infty$.
\end{lem}

\begin{proof}
We fix a test function $\psi \in C_{0}^{\infty}(\C)$ and let $\psi_n (\zeta) = \psi (r_n^{-1} \zeta)$.
The change of variables $z=r_n^{-1}\zeta$ and $w= r_n^{-1}\eta$ gives that
\begin{align*}
	\E_n I_n[\psi_n] &= \int_{\C} \psi_n(\zeta)\, dA(\zeta) \int_{\C} \frac{\mathbf{R}_{n,2}(\zeta, \eta)}{\zeta- \eta} \,dA(\eta)\\
	&= r_n^{-1} \int_{\C} \psi (z)\, dA(z) \int_{\C} \frac{R_{n,2}(z,w)}{z-w}\, dA(w)
\end{align*}
and
\begin{align*}
	\E_n II_n[\psi_n] & = n \int_{\C} \partial Q(\zeta)\, \psi_n(\zeta)\, \mathbf{R}_{n,1}(\zeta)\, dA(\zeta)
	= n\int_{\C} \partial Q( r_n z)\, \psi(z)\, R_{n,1}(z)\, dA(z).
\end{align*}
Likewise, changing variables and integrating by parts, we obtain
\begin{align*}
	\E_n III_n[\psi_n] &= \int_{\C} \partial\psi_n(\zeta) \, \mathbf{R}_{n,1}(\zeta)\, dA(\zeta)
	= r_n^{-1} \int_{\C} \partial \psi (z)\, R_{n,1}(z)\, dA(z)\\
	&= -r_n^{-1} \int_{\C} \psi(z) \,\partial R_{n,1}(z)\, dA(z).
\end{align*}
Hence, by the Ward identity in \eqref{W:identity}, we have
\begin{align*}
	&\int_{\C} \psi(z)\, dA(z) \int_{\C} \frac{R_{n,2}(z,w)}{z-w} \,dA(w) \\
	&= n\, r_n \int_{\C} \partial Q (r_n z)\, \psi(z) \, R_{n,1}(z) \, dA(z) + \int_{\C} \psi(z)\, \partial R_{n,1} (z) \, dA(z).
\end{align*}
Since $\psi$ is an arbitrary test function, we have in the sense of distributions,
\begin{align*}
	\int_{\C} \frac{R_{n,2}(z,w)}{z-w} \, dA(w)
	= n r_n \partial Q(r_n z) \, R_{n,1}(z) + \partial R_{n,1}(z).
\end{align*}
Dividing through by $R_{n,1}(z)$ and using the fact that
\begin{align*}
	R_{n,2}(z,w) = R_{n,1}(z) \left( R_{n,1}(w) - B_{n}(z,w)  \right),
\end{align*}
we obtain
\begin{align*}
	\int_{\C} \frac{R_{n,1}(w)}{z-w} dA(w) - \int_{\C} \frac{B_n(z,w)}{z-w} dA(w)
	= n r_n \partial Q(r_n z) + \partial \log R_{n,1}(z).
\end{align*}
Differentiating with respect to $\bar{z}$, we get
\begin{align*}
	R_{n,1}(z) - \bar{\partial}C_{n}(z) = n r_n^{\,2}\, \Lap Q(r_n z) + \Lap \log R_{n,1}(z).
\end{align*}
Since $\Lap Q( r_n z) = r_n^{\,2(k-1)} \Lap Q_0(z) + O(r_n^{\,2k-1})$ uniformly on compact subsets of $\C$  as $n \to \infty$ and $r_n = n^{-1/2k}$ we obtain
\begin{align*}
	\bar{\partial} C_n(z) = R_{n,1}(z)- \Lap Q_0(z) - \Lap \log R_{n,1}(z) + o(1)
\end{align*}
where $o(1) \to 0$ uniformly on compact subsets of $\C$ as $n \to \infty$.
\end{proof}

\subsection{The proof of Theorem \ref{01ward}} We will need a few lemmas.

\begin{lem} \label{R:zero}
If $R(z_0)=0$ then there is a real analytic function $\tilde{R}$ such that
\begin{align*}
	R(z)= \babs{\,z-z_0\,}^{\,2} \tilde{R}(z).
\end{align*}
If $R$ does not vanish identically, then all zeros of $R$ are isolated.
\end{lem}

\begin{proof}
The assumption gives that the holomorphic kernel $L$ corresponding to $R$ satisfies $L(z_0, z_0)=0$. Hence
$\int e^{-Q_0(w)}\,|\,L(z_0,w)\,|^{\,2} dA(w) \leq 0$ by the mass-one inequality \eqref{mass:ineq}.
Thus $L(z_0,w)=0$ for all $w\in\C$. Since $L$ is Hermitian-entire, we can thus write
$$L(z,w)=(z-z_0)\,(w-z_0)^*\,\tilde{L}(z,w)$$ for some Hermitian-entire function $\tilde{L}$.
We now have $R(z)=\babs{\,z-z_0\,}^{\,2} \tilde{L}(z,z)e^{-Q_0(z)}$.

For the second statement, we assume that $R$ does not vanish identically and there exists a zero $z_0$ of $R$ which is not isolated. Then, we can take a sequence $\{z_j\}_{1}^{\infty}$ of distinct zeros of $R$ which converges to $z_0$, whence by the above argument, for each $j$ we obtain $L(z_j, w)=0$ for all $w \in \C$. If we fix $w$, then $L(z,w)=0$ for all $z\in \C$ since $L(z,w)$ is holomorphic in $z$. Hence $L=0$ identically.
\end{proof}

\begin{lem}\label{lem:logsub}
$L(z,w)$ is a positive matrix and $z \mapsto L(z,z)$ is logarithmically subharmonic.
\end{lem}

\begin{proof}
  It is clear that $L$ is a positive matrix. Now write $L_z(w) : = L(w,z)$ and define a semi-definite inner product by $\langle L_z, L_w \rangle_{*} :=L(w,z)$ on the linear span of the functions $L_z$ for $z\in\C$. The completion of this span forms a (perhaps semi-normed) Hilbert space $\calH_{*}$ and $L$ is a reproducing kernel of the space. Now when $L(z,z)>0$
\begin{align}\label{Lap:log}
	\Lap_z \log L(z,z) = \frac{L(z,z) \Lap L(z,z)-\partial_z L(z,z)\, \bar{\partial}_z L(z,z)}{L(z,z)^2}.
\end{align}
Since $L(z,w)$ is Hermitian-entire, we have $\bar{\partial}_{z}L_z \in \calH_{*}$, $\langle \bar{\partial}_z L_z, L_z \rangle_{*} = \bar{\partial}_z L(z,z)$, and
$\langle \bar{\partial}_z L_z , \bar{\partial}_z L_z \rangle_{*}= \Lap L(z,z)$. Hence, the numerator of (\ref{Lap:log}) can be written as
\begin{align*}
	\|\, L_z\, \|_{*}^{\,2} \cdot \|\,\bar{\partial} L_z\, \|_{*}^{\,2} - \babs{\,\langle \bar{\partial}_z L_z, L_z \rangle_{*}\,}^{\,2},
\end{align*}
which is non-negative by the Cauchy-Schwarz inequality.

At points where $L(z,z)=0$, $\log L(z,z)$ satisfies the sub-mean value property since $\log L(z,z)=-\infty$. Hence the function $\log L(z,z)$ is subharmonic on $\C$.
\end{proof}

\begin{lem} \label{zjick}
If $R(z_0)=0$ and $R(z)=\babs{\,z-z_0\,}^{\,2}\tilde{R}(z)$, then $\Lap Q_0 + \Lap \log \tilde{R} \geq 0$ in a neighborhood of $z_0$.
\end{lem}
\begin{proof}
We choose a small disc $D=D(z_0, \epsilon)$ and consider the function $$S(z) = \log \left(e^{\,Q_0(z)}\,\tilde{R}(z)\right).$$ Observing that $\Lap_z \log L(z,z) = \Lap Q_0(z) + \Lap \log \tilde{R}(z) + \delta_{z_0}$ in the sense of distributions, Lemma \ref{lem:logsub} gives us that $\Lap S \geq 0$ in the sense of distributions on $D\backslash \{z_0\}$. If $\tilde{R}(z_0)>0$
we extend $S$ analytically to $z_0$. On the other hand, if $\tilde{R}(z_0)=0$ we define $S(z_0)=-\infty$. In both cases, the extended function $S$ is subharmonic on $D$.
\end{proof}

We now turn to the left hand side in the rescaled version of Ward's identity, namely the function
$\dbar C_n$ where $C_n$ is the Cauchy transform of $B_n$ (see \eqref{cauchyn}).

\begin{lem}\label{ubdd} Suppose that $R=\lim R_{n_\ell}$ is a limiting $1$-point function which does not vanish identically.
Let $Z$ be the set of isolated zeros of $R$ and let $B(z,w)=\lim B_{n_\ell}(z,w)$ be the corresponding Berezin kernel for $z\not\in Z$.
Then $C_{n_\ell} \to C$ locally uniformly on the complement $Z^{c}=\C\setminus Z$ as $\ell\to \infty$, where the function $$C(z)=\int\frac {B(z,w)}{z-w}\, dA(w)$$
is bounded on $Z^{c} \cap \compset$ for each compact subset $\compset$ of $\C$.
\end{lem}
\begin{proof}
 We have that $c_{n_\ell}K_{n_\ell} \to K$ locally uniformly on $\C^2$ where $K(z,z)=R(z)>0$ when $z\not\in Z$. Hence, for fixed $\epsilon$ with $0<\epsilon<1 $ we can choose $N$ such that if $\ell \geq N$ then
\begin{align*}
	\babs{\,B_{n_\ell}(z,w) - B(z,w)\,} < \epsilon^{\,2}
\end{align*}
for all $z,w$ with $\babs{\,z\,} \leq 1/\epsilon$, $\babs{\,w\,} \leq 2/\epsilon$, and $\dist(z, Z) \ge \epsilon$.
Then, for $z$ with $\babs{\,z\,} \leq 1/\epsilon$ and $\dist(z, Z) \ge \epsilon$,
\begin{align*}
	\babs{\,C_{n_\ell}(z) - C(z)\,}
	&\leq \left( \int_{\babs{\,z-w\,}<1/\epsilon}+ \int_{\babs{\,z-w\,}>1/\epsilon} \right)
	\babs{\,\frac{B_{n_\ell}(z,w) - B(z,w)}{z-w}\,} dA(w)\\
	&\leq \epsilon^{\,2} \int_{\babs{\,z-w\,}<1/\epsilon} \frac{1}{\babs{\,z-w\,}} dA(w)
	+ \epsilon \int \babs{\,B_{n_\ell}(z,w)-B(z,w)\,} dA(w)\\
	&\leq 4 \epsilon.
\end{align*}
Here, we have used the mass-one inequality for the third inequality. Thus $C_{n_\ell} \to C$ uniformly on compact subsets of $Z^c$.

Now fix a compact subset $\compset$ of $\C$. Then, for all $z,w$ with $z \in \compset\setminus Z$ and $\dist(w,\compset)\leq 1$
\begin{align*}
	B_{n_\ell}(z,w) = \frac{\babs{\,K_{n_\ell}(z,w)\,}^{\,2}}{K_{n_\ell}(z,z)} \leq K_{n_\ell}(w,w) \leq M
\end{align*}
for some $M=M_\compset$ which depends only on $\compset$ by Lemma \ref{lem:bound}.
Thus, for $z \in \compset \setminus Z$
\begin{align*}
	\babs{C_{n_\ell}}
	&\leq \left( \int_{\babs{\,z-w\,}<1}+ \int_{\babs{\,z-w\,}>1} \right)
	\babs{\,\frac{B_{n_\ell}(z,w)}{z-w}\,} dA(w)\\
	& \leq M \int_{\babs{\,z-w\,}<1} \frac{1}{\babs{\,z-w\,}} dA(w) +  \int B_{n_\ell}(z,w) dA(w) \leq 2M+1
\end{align*}
Hence we obtain $\babs{\,C(z)\,} \leq 2M+1$ for $z \in \compset \setminus Z$.
\end{proof}

\begin{lem}
If $R$ does not vanish identically, the Ward's equation
\begin{equation}\label{wa:2}
	\bar{\partial} C = R - \Lap Q_0 - \Lap \log R
\end{equation}
holds in the sense of distributions.
\end{lem}

\begin{proof} The preceding lemmas show that
\begin{equation}\label{rwa}\dbar C_n=R_n-1-\Lap\log R_n+o(1)\end{equation} and that a subsequence $C_{n_\ell}$ converges to $C$ boundedly and locally uniformly on $\C\setminus Z$. Since $Z\cap \compset$ is a finite set
for each compact set $\compset$, it follows that $C_{n_\ell}\to C$ in the sense of distributions, and hence $\dbar C_{n_\ell}\to\dbar C$. By Ward's equation and the locally uniform convergence $R_{n_\ell}\to R$ it then follows that
$\Lap\log R_{n_\ell}\to\Lap\log R$ in the sense of distributions. We can thus pass to the limit as $n_\ell\to\infty$
in the rescaled Ward identity \eqref{rwa}.
\end{proof}

\begin{proof}[Proof of Theorem \ref{01ward}] We follow the strategy in \cite{AKM}, Theorem 4.8.
Suppose that $R(z_0)=0$. We must prove that $R=0$ identically.

Let $D$ be a small disk centered at $z_0$ and write $\chi=\chi_D$ for the characteristic function. Also write $R(z)=\babs{\,z-z_0\,}^{\,2}\tilde{R}(z)$.

Consider the measures
$\mu=\chi\cdot(\Lap Q_0+\Lap\log R)$ and $\nu=\chi\cdot(\Lap Q_0+\Lap\log\tilde{R})$. By lemmas \ref{lem:logsub} and \ref{zjick}, these measures are positive, and $\mu=\delta_{z_0}+\nu$.
Write $C^\mu(z)=\int_\C\frac 1 {z-w}\, d\mu(w)$ for the Cauchy transform of $\mu$. Clearly,
$$C^\mu(z)=\frac 1 {z-z_0}+C^\nu(z),\quad z\in D.$$
Also $\dbar C^\nu=\nu\ge 0$. When $z\in D$, the right hand side in Ward's equation equals $R(z)-\Lap (Q_0+\log R)(z)=R(z)-\dbar C^\mu(z).$
If $C(z)=\int\frac {B(z,w)}{z-w}\, dA(w)$, we have, by Ward's equation, that
$$\dbar(C+C^\mu)(z)=R(z).$$
Hence, by Weyl's lemma, $C(z)=-1/(z-z_0)-C^\nu(z)+v(z)$ where $v$ is smooth near $z_0$. If $C^\mu(z)$ were bounded as $z\to z_0$ then the measure $\mu=\nu+\delta_{z_0}$ would place no mass at $\{z_0\}$, so
$\nu=-\delta_{z_0}+\rho$ where $\rho(\{z_0\})=0$. This contradicts that $\nu\ge 0$. The contradiction shows that
$\babs{\,C(z)\,}\to\infty$ as $z\to z_0$. This in turn contradicts that $C$ is bounded (Lemma \ref{ubdd}), and hence $R(z_0)=0$ is impossible. Hence $\Lap\log R$ is a smooth function on $\C$. Applying Weyl's lemma to the distributional Ward equation
$\dbar C=R-\Lap Q_0-\Lap\log R$ now shows that $C(z)$ is smooth and hence that the equation holds pointwise on $\C$.
\end{proof}

\section{Universality results}\label{Sec4}
In this section, we prove theorems \ref{mth2} and \ref{mth3}. The proof of Theorem \ref{mth3} relies on certain apriori estimates, whose proofs are postponed to Section \ref{Sec5}.

\subsection{Homogeneous singularities}
Assume that $Q$ has a homogeneous singularity of type $2k-2$ at the origin, i.e., that the canonical decomposition is of the form
$Q=Q_0+\re H,\, H=c\,\zeta^{\,2k}$, where $Q_0$ is positively homogeneous of degree $2k$. As always, we write
$\mu_0$ for the measure $d\mu_0=e^{-Q_0}\, dA$.

We now recall the kernel $L_n$ (defined in \eqref{rhk})
$$L_n(z,w)=k_n(z,w)e^{-H_n(z)/2-\bar{H}_n(w)/2},\quad (H_n(z)=nH(r_nz),\quad r_n=n^{-1/2k}).$$
In the present case,
$L_n(z,w)=k_n(z,w)e^{-c\,z^{\,2k}/2-\bar{c}\,\bar{w}^{\,2k}/2}.$
By Lemma \ref{lem:compact}, $L_n$ is the reproducing kernel for the space
$$\calH_n=\{f(z)=q(z)\cdot e^{-c\,z^{\,2k}/2};\,q\in\Poly(n)\}$$
regarded as a subspace of $L^2(\mu_0)$. (This is because $\tilde{\mu}_n=\mu_0$ for the measure $\tilde{\mu}_n$ in \eqref{apmeas}.)

Since the spaces $\calH_n$ are increasing,
$\calH_n\subset\calH_{n+1},$
where the inclusions are isometric, it follows that a unique limiting holomorphic kernel $L=\lim L_{n}$ exists.
 By Theorem \ref{pthm}, the kernel $L$ is the reproducing kernel for a contractively embedded subspace $\calH_*$ of $L^2_a(\mu_0)$, which must contain the dense subset $U= \bigcup \calH_n$. Furthermore, by the reproducing property of $L_n$, we have for each element $f(z)=q(z)\cdot e^{-cz^{2k}/2}\in U$ that
$\langle f,L_{n,z}\rangle_{L^2(\mu_0)}=f(z)$, whenever $n>\operatorname{degree}q$.
It follows that
$$f(z)=\lim_{n\to\infty}\langle f,L_{n,z}\rangle_{L^2(\mu_0)}=\langle f,L_z\rangle_{L^2(\mu_0)},\quad f\in U.$$
Since $U$ is dense in $L^2_a(\mu_0)$, $L$ must equal to the reproducing kernel $L_0$ of $L^2_a(\mu_0)$.
The proof of Theorem \ref{mth2} is complete. q.e.d.

\subsection{Rotational symmetry}
Referring to the canonical decomposition $Q=Q_0+\re H+Q_1$ we now
suppose that $Q_0(z)=Q_0(|\,z\,|)$, and we fix a rotationally symmetric limiting holomorphic kernel
$$L(z,w)=E(z\bar{w}).$$ Writing $E(z)=\sum_0^\infty a_jz^j$, the mass-one inequality
$$\int e^{-Q_0(w)}\babs{\,L(z,w)\,}^{\,2}\, dA(w)\le L(z,z)$$
is seen to be equivalent to that
\begin{equation}\label{co:mass}\sum \babs{\,a_j\,}^{\,2}\babs{\,z\,}^{\,2j}\|\,w^{\,j}\,\|_{L^2(\mu_0)}^{\,2}\le \sum a_j\babs{\,z\,}^{\,2j}.\end{equation}
To use Ward's equation, we first compute the Cauchy transform $C(z)$ as follows:
\begin{align*}C(z)&=\frac 1 {L(z,z)}\int_\C\frac {e^{-Q_0(w)}}{z-w}\babs{\,L(z,w)\,}^{\,2}\, dA(w)\\
&=\frac 1 {E(\babs{\,z\,}^{\,2})}\sum_{j,k}a_j\bar{a}_kz^{\,j}\bar{z}^{\,k}\int_\C \frac {e^{-Q_0(w)}}{z-w}\bar{w}^{\,j}w^{\,k}\, dA(w)\\
&=\frac 1 {E(|\,z\,|^{\,2})}\sum_{j,k}a_j\bar{a}_kz^{\,j}\bar{z}^{\,k}\frac 1 \pi \int_0^\infty
e^{-Q_0(r)}r^{\,j+k}\, dr\int_0^{2\pi}\frac {e^{\,i(k-j)\theta}}{z/r-e^{\,i\theta}}\, d\theta.
\end{align*}
However, as is shown in \cite{AK}, we have that
$$\frac 1 {2\pi}\int_0^{2\pi}\frac {e^{\,i(k-j)\theta}}{z/r-e^{\,i\theta}}\, d\theta=
\begin{cases}-(z/r)^{\,k-j-1}&\,\,\text{if}\,\,|\,z\,|<r,\,\, k-j\ge 1,\cr
\,\,\,\,(z/r)^{\,k-j-1}&\,\,\text{if}\,\,|\,z\,|>r,\,\,k-j\le 0,\cr
\quad 0&\,\,\text{otherwise}.\cr
\end{cases}$$
Thus
\begin{align*}C(z)&=\frac 2 {E(|\,z\,|^{\,2})}\sum_{j,k}a_j\bar{a}_kz^{\,j}\bar{z}^{\,k}\\
&\left(\int_0^{|\,z\,|}e^{-Q_0(r)}r^{j+k}\left(\frac z r\right)^{k-j-1}\1(k\le j)\, dr-\int_{|\,z\,|}^\infty
e^{-Q_0(r)}r^{j+k}\left(\frac z r\right)^{k-j-1}\1(k\ge j+1)\, dr\right)\\
&=A(z)-B(z)\end{align*}
where
\begin{align*}A(z)&=\frac 2 {E(|\,z\,|^{\,2})}\sum_{j,k}a_j\bar{a}_kz^{\,j}\bar{z}^{\,k}\int_0^{|\,z\,|}e^{-Q_0(r)}r^{\,j+k}\left(\frac z r\right)^{\,k-j-1}\, dr,\\ B(z)&=\frac 2 {E(|\,z\,|^{\,2})}\sum_{j,k}a_j\bar{a}_kz^{\,j}\bar{z}^{\,k}\int_{0}^\infty
e^{-Q_0(r)}r^{\,j+k}\left(\frac z r\right)^{\,k-j-1}\1(k\ge j+1)\, dr.
\end{align*}
The term $A(z)$ can be written as
\begin{align*}A(z)&=\frac 1 {zE(|\,z\,|^{\,2})}\sum_{j,k}a_j\bar{a}_k\babs{\,z\,}^{\,2k}\int_0^{|\,z\,|^{\,2}}e^{-Q_0( \sqrt{r})}r^{\,j} \, dr\\
&=\frac 1 z\int_0^{|\,z\,|^{\,2}}e^{-Q_0(\sqrt{r})} E(r)\, dr,\end{align*}
which gives
$$\dbar A(z)=e^{-Q_0(z)}E(|\,z\,|^{\,2})=R(z).$$
The term $B(z)$ is computed as follows,
\begin{align*}B(z)&=\frac 1 {E(|\,z\,|^{\,2})}\sum_{k=1}^\infty\bar{a}_kz^{\,k-1}\bar{z}^{\,k}\sum_{j=0}^{k-1} a_j\int_0^\infty
e^{-Q_0(\sqrt{r})}r^{\,j}\, dr\\
&=\frac 1 {E(|\,z\,|^{\,2})}\sum_{k=1}^\infty\bar{a}_kz^{\,k-1}\bar{z}^{\,k}\sum_{j=0}^{k-1}a_j\|\,z^{\,j}\,\|_{Q_0}^{\,2}.
\end{align*}
Noting that
$$\d_z\log L(z,z)=\frac {\d_z E(|\,z\,|^{\,2})}{E(|\,z\,|^{\,2})}=\frac 1 {E(|\,z\,|^{\,2})}\sum_{k=1}^\infty k\,\bar{a}_kz^{\,k-1}\bar{z}^{\,k},$$
we infer that Ward's equation
$$\dbar A-\dbar B=R-\Lap_z\log L(z,z)$$
is equivalent to that $\dbar (B-\d_z\log L(z,z))=0$. This in turn, is equivalent to that the function
$$\frac 1 {E(|\,z\,|^{\,2})}\sum_{k=1}^\infty \bar{a}_kz^{\,k-1}\bar{z}^{\,k}\left(k-\sum_{j=0}^{k-1}a_j\|\,z^{\,j}\,\|_{L^2(\mu_0)}^{\,2}\right)$$
be entire. It is easy to check that this is the case if and only if all coefficients in the sum vanish, that is, if and only if for each $k\geq 1$ we have that
\begin{equation}\label{01}a_k=0\quad \text{or}\quad \sum_{j=0}^{k-1}a_j\|\,z^{\,j}\,\|_{L^2(\mu_0)}^{\,2}=k.\end{equation}

We now apply the growth estimate in Theorem \ref{apth}, part \ref{r1p}, which says that
\begin{equation}\label{grow}E(|\,z\,|^{\,2})=\Lap Q_0(z)\,e^{\,Q_0(z)}\,(1+o(1))\quad\text{as}\quad  z\to\infty.\end{equation}
We claim that this implies the second alternative in \eqref{01}.

 Indeed, \eqref{grow} is clearly not satisfied if $E$ is constant. Next note that the mass-one inequality \eqref{co:mass} and the zero-one law (Theorem \ref{01ward}) imply that $0<a_0 \leq 1/\|\,1\,\|^{\,2}_{L^2(\mu_0)}$. Since $E(z)$ is not a polynomial by \eqref{grow}, for any $k$ there exists $N \geq k$ such that $a_N \neq 0$. By \eqref{01}, we obtain that if $a_N \neq 0$ but $a_j=0$ for all $j$ with $1\leq j \leq N-1$ then $N=1$ and $a_0 = 1/\|\,1\,\|^{\,2}_{L^2(\mu_0)}$. By a simple induction, we then have $a_k=1/\|\,z^{\,k}\,\|_{L^2(\mu_0)}^{\,2}$ for all $k \geq 0$. Thus, we have
$$E(z)=\sum_{j=0}^\infty\frac 1 {\|\,z^{\,j}\,\|_{L^2(\mu_0)}^{\,2}}z^{\,j}.$$
Since the polynomial $\phi_j(z)=z^{\,j}/\|\,z^{\,j}\,\|_{L^2(\mu_0)}$ is the $j$:th orthonormal polynomial with respect to the measure $\mu_0$, we have
$$L(z,w)=E(z\bar{w})=\sum_{j=0}^\infty \phi_j(z)\bar{\phi}_j(w)=L_0(z,w)$$
where $L_0$ is the Bergman kernel for the space $L^2_a(\mu_0).$
The proof is complete. q.e.d.

\section{Asymptotics for $L_0(z,z)$} \label{Sec4.5}
In this section, we prove part \ref{r0p} of Theorem \ref{apth}.

To this end, let $A_0(z,w)$ be the Hermitian polynomial such that $A_0(z,z)=Q_0(z)$ and
put
$$L_0^\sharp(z,w)=\left[\d_1\dbar_2 A_0\right](z,w)\cdot e^{\,A_0(z,w)}.$$
We write $L_z^\sharp(w)$ for $L_0^\sharp(w,z)$ and, for suitable functions $u$,
$$\pi^\sharp u(z)=\langle u,L_z^\sharp\rangle_{L^2(\mu_0)}=\int_\C u \bar{L}_z^\sharp e^{-Q_0} \, dA.$$
Below, we fix a $z$ with $|z|$ large enough; we must estimate $L_0(z,z)$.
We also fix a number $\delta_0=\delta_0(z)>0$ and write $\chi_z$ for a fixed $C^\infty$-smooth test-function
with $\chi_z(w)=1$ when $|\,w-z\,|\le \delta_0$ and $\chi_z(w)=0$ when $|\,w-z\,|\ge 2\delta_0$.

We will use the following estimate.

\begin{lem} \label{abo} If $|1-w/z|$ is sufficiently small, then
$2\re A_0(z,w)\le Q_0(z)+Q_0(w)-c|z|^{2k-2}|w-z|^2$ where $c$ is a positive constant.
\end{lem}

\begin{proof} Put $h=w-z$. By Taylor's formula, $A_0(w,z)=Q_0(z)+\sum_1^{2k}\frac {\d^jQ_0(z)}{j!}h^j$. Similarly,
$A_0(w,w)=Q_0(z)+\sum_{i+j\ge 1}\frac {\d^i\dbar^j Q_0(z)}{i!j!}h^i\bar{h}^j.$ Hence
\begin{equation}\label{broj}2\re A_0(z,w)-Q_0(z)-Q_0(w)+\Lap Q_0(z)|h|^2=-\sum_{i,j\ge 1,i+j\ge 3}\frac {\d^i\dbar^j Q_0(z)}
{i!j!}h^i\bar{h}^j.\end{equation}
However, since $Q_0$ is homogeneous of degree $2k$, the derivative $\d^i\dbar^j Q_0$ is homogeneous of degree
$2k-i-j$. Hence
$$\babs{\d^i\dbar^j Q_0(z)}\babs{w-z}^{i+j}\le C|z|^{2k-2}|w-z|^2|1-w/z|^{i+j-2}.$$
Thus, if $i+j\ge 3$ and $|1-w/z|$ is sufficiently small, then the left hand side in \eqref{broj} is dominated by an arbitrarily small multiple of $|z|^{2k-2}|z-w|^2$. On the other hand, by homogeneity and positive definiteness of $\Lap Q_0$ we have that
$\Lap Q_0(z)|z-w|^2\ge c'|z|^{2k-2}|z-w|^2$ where $c'$ is a positive constant. The lemma thus follows with any positive constant $c<c'$
\end{proof}

As always, we write $d\mu_0=e^{-Q_0}\, dA$; $L^2_a(\mu_0)$ denotes the associated Bergman space of entire functions, and $L_0$ is the Bergman kernel of that space.

\begin{lem} \label{uv}
Let $|z|\ge 1$ and $\delta_0$ a positive number with $\delta_0/|z|$ sufficiently small. Then there is a constant $C=C(\delta_0)$ such that, for all functions $u\in L^2_a(\mu_0)$
$$\babs{\,u(z)-\pi^\sharp[\chi_zu](z)\,}\le C
\|\,u\,\|_{L^2(\mu_0)}\,(\delta_0^{-1}+1)e^{\,Q_0(z)/2}.$$
\end{lem}

\begin{proof} Note that
\begin{equation}\label{pa:1}\begin{split}\pi^\sharp[\chi_zu](z)&=\int_\C\chi_z(w)u(w)\left[\d_1\dbar_2A_0\right](z,w)\cdot e^{\,A_0(z,w)-A_0(w,w)}\, dA(w)\\
&=-\int_\C\frac {u(w)\chi_z(w)F(z,w)}{w-z}\dbar_w\left[e^{\,A_0(z,w)-A_0(w,w)}\right]\, dA(w),\\
\end{split}
\end{equation}
where
\begin{equation}\label{fzw}F(z,w)=\frac {(w-z)\left[\d_1\dbar_2 A_0\right](z,w)}{\dbar_2A_0(w,w)-\dbar_2A_0(z,w)}.\end{equation}
Now fix $w$. The denominator $P(z)=\dbar_2A_0(w,w)-\dbar_2A_0(z,w)$ is by Taylor's formula equal to the polynomial
$$- \Lap Q_0(w)\cdot(z-w) - \frac {\d\Lap Q_0(w)} 2\cdot (z-w)^{\,2} - \cdots- \frac {\d^{k-1}\Lap Q_0(w)}{k!}\cdot (z-w)^{\,k}.$$
Here the derivative
$\d^j\Lap Q_0(w)=\babs{\,w\,}^{\,2k-2-j} \d^j\Lap Q_0(w/|\,w\,|)$
is positively homogeneous of degree $2k-2-j$. Put $c(w)=\Lap Q_0(w/|\,w\,|)$. We then have that
$$P(z)=c(w)|\,w\,|^{\,2k-2}\cdot(w-z)+O(\,(w-z)^{\,2}\,),\quad (z\to w).$$
Since also $\d_1\dbar_2A_0(z,w)=c(z)\babs{\,z\,}^{\,2k-2}(1+O(w-z))$, we have by \eqref{fzw}
\begin{equation}\label{eu}F(z,w)=1+O(w-z),\qquad (w\to z).\end{equation}
By the form of $F$ it is also clear that
\begin{equation}\label{fu}\dbar_2 F(z,w)=O(w-z),\quad (w\to z).\end{equation}

An integration by parts in \eqref{pa:1} gives
$\pi^\sharp[\chi_zu](z)=u(z)+\epsilon_1+\epsilon_2$ where
\begin{align*}\epsilon_1&=\int\frac {u(w)\dbar\chi_z(w)F(z,w)}{w-z}\,e^{\,A_0(z,w)-A_0(w,w)}\, dA(w),\\
\epsilon_2&=\int\frac {u(w)\chi_z(w)\dbar_2F(z,w)}{w-z}\,e^{\,A_0(z,w)-A_0(w,w)}\, dA(w).
\end{align*}
Inserting the estimates \eqref{eu} and \eqref{fu}, using also that $\dbar\chi_z(w)=0$ when $|\,w-z\,|\le\delta_0$ we find that
\begin{align*}\babs{\,\epsilon_1\,}&\le C\delta_0^{-1}\int\babs{\,u(w)\,}\babs{\,\dbar\chi_z(w)\,}e^{\,\re A_0(z,w)-Q_0(w)}\, dA(w),\\
\babs{\,\epsilon_2\,}&\le C\int\chi_z(w)\babs{\,u(w)}e^{\,\re A_0(z,w)-Q_0(w)}\,dA(w).
\end{align*}
To estimate $\epsilon_1$ we use Lemma \ref{abo} to get
\begin{equation}\label{tay}e^{\,\re A_0(z,w)-Q_0(w)/2}\le Ce^{\,Q_0(z)/2-c|z-w|^2}.\end{equation}
This gives
\begin{align*}\babs{\,\epsilon_1\,}e^{-Q_0(z)/2}&\le C\delta_0^{-1}\int\babs{\,u(w)\,}|\,\dbar\chi_z(w)\,|e^{-Q_0(w)/2}\, dA(w)\\
&\le
C\delta_0^{-1}\|\,u\,\|_{L^2(\mu_0)} \|\,\dbar\chi_z\,\|_{L^2} \le C'\|\,u\,\|_{L^2(\mu_0)}.
\end{align*}
To estimate $\epsilon_2$ we note that (again by \eqref{tay})
$$\babs{\,\epsilon_2\,}e^{-Q_0(z)/2}\le C\|\,u\,\|_{L^2(\mu_0)}\left(\int_{|\,w-z\,|\le 2\delta_0}e^{-c|z-w|^2}dA(w)\right)^{1/2}\le C\|\,u\,\|_{L^2(\mu_0)}.$$
The proof is complete.
\end{proof}

Let $\pi_0:L^2(\mu_0)\to L^2_a(\mu_0)$ be the Bergman projection, $\pi_0[f](z)=\langle f,L_z\rangle_{L^2(\mu_0)}$, where we write $L_z(w)$ for $L_0(w,z)$. Noting that
$$(\pi^\sharp[\chi_zL_z](z))^{\,*}=\langle\chi_zL_z,L_z^\sharp\rangle^{\,*}=\langle \chi_zL_z^\sharp,L_z\rangle=\pi_0[\chi_zL_z^\sharp](z),$$
we see that
$$\babs{\,L_z(z)-\pi_0[\chi_z L_z^\sharp](z)\,}=\babs{\,L_z(z)-\pi^\sharp[\chi_zL_z](z)\,}.$$
If we now choose $u=L_z$ in Lemma \ref{uv} and recall that $\|\,L_z\,\|_{L^2(\mu_0)}^{\,2}=L_0(z,z)$, we obtain the estimate
\begin{equation}\label{l0:cor}\babs{\,L_0(z,z)-\pi_0[\chi_zL_z^\sharp](z)\,}\le C\sqrt{L_0(z,z)}\cdot e^{\,Q_0(z)/2},\qquad  \babs{z}\ge 1.\end{equation}

\begin{lem}\label{ho:lem} There is a constant $C$ such that for all $|z|\ge 1$ and all $\delta_0=\delta_0(z)>0$ with $\delta_0/|z|$ small enough
$$\babs{\,\Lap Q_0(z)\,e^{\,Q_0(z)}-\pi_0\left[\chi_zL_z^\sharp\right](z)\,}\le C|\,z\,|^{\,k-1}e^{\,Q_0(z)}.$$
\end{lem}

\begin{proof} Consider the function
$u_0=\chi_zL_z^\sharp-\pi_0[\chi_zL_z^\sharp].$
This is the norm-minimal solution in $L^2(\mu_0)$ to the problem $\dbar u=(\dbar\chi_z)\cdot L_z^\sharp$.

Since $Q_0$ is strictly subharmonic on the support of $\chi_z$ we can apply the standard Hörmander estimate (e.g. \cite{H}, p. 250) to obtain
\begin{align*}\|\,u&\,\|_{L^2(\mu_0)}^{\,2}\le\int_\C\babs{\,\dbar\chi_z\,}^{\,2}\babs{\,L_z^\sharp\,}^{\,2}\frac {e^{-Q_0}}{\Lap Q_0}\, dA\\
&\le C \babs{\,z\,}^{\,-(2k-2)}\|\,\dbar \chi_z\,\|_{L^2}^{\,2}\sup_{\delta_0\le \babs{w-z}\le 2\delta_0}
\babs{\,[\d_1\dbar_2 A_0](z,w)\,}^{\,2}e^{\,2\re A_0(z,w)-A_0(w,w)},
\end{align*}
where we used homogeneity of $\Lap Q_0$.

By Taylor's formula and the estimate \eqref{tay} we have when $\delta_0\le \babs{\,w-z\,}\le 2\delta_0$
$$\babs{\,[\d_1\dbar_2 A_0](z,w)\,}^{\,2}e^{\,2\re A_0(z,w)-A_0(w,w)}\le C\Lap Q_0(z)^{\,2}e^{\,Q_0(z)-2c|z|^{2k-2}|z-w|^2}.$$
By the homogeneity of $\Lap Q_0$
we thus obtain the estimate
\begin{equation}\label{verd}\|\,u\,\|_{L^2(\mu_0)}\le C|\,z\,|^{\,k-1} e^{\,Q_0(z)/2-c'\delta_0^2|z|^{2k-2}}.\end{equation}
 We now pick another (small) number $\delta>0$ and
 invoke the following pointwise-$L^2$ estimate (see e.g \cite{AKM}, Lemma 3.1, or the proof of the inequality \eqref{sub:ineq})
\begin{equation}\label{vemm}\babs{\,u(z)\,}^{\,2}e^{-Q_0(z)}\le Ce^{c''\delta\Lap Q_0(z)|z|}\delta^{-2}\int_{D(z,\delta)}\babs{\,u(w)\,}^{\,2}e^{-Q_0(w)}\, dA(w).\end{equation}
Combining with \eqref{verd}, this gives
\begin{equation}\label{veck}\babs{\,u(z)\,}^{\,2}e^{-Q_0(z)}\le C\delta^{-2}e^{-c'\delta_0^2|z|^{2k-2}+c''\delta|z|^{2k-1}}|\,z\,|^{\,2k-2}e^{\,Q_0(z)}.\end{equation}
Choosing $\delta_0$ a small multiple of $|z|^{1/2}$ and then $\delta$ small enough, we insure that the right hand side is
dominated by $C|z|^{2k-2}e^{Q_0(z)}$, as desired.
\end{proof}

\begin{proof}[Proof of Part \ref{r0p} of Theorem \ref{apth}]
By the estimate \eqref{l0:cor} and Lemma \ref{ho:lem} we have
\begin{align*}\babs{\,\Lap Q_0(z)\,e^{\,Q_0(z)}-L_0(z,z)\,}&\le C_1\sqrt{L_0(z,z)}e^{\,Q_0(z)/2}+C_2|\,z\,|^{\,k-1}e^{\,Q_0(z)}.
\end{align*}
Writing $R_0(z)=L_0(z,z)e^{-Q_0(z)}$, this becomes
\begin{equation}\label{s:in}\babs{\,|\,z\,|^{\,2k-2}c(z)-R_0(z)\,}\le C_1\sqrt{R_0(z)}+C_2|\,z\,|^{\,k-1},\quad c(z)=\Lap Q_0(z/|\,z\,|).\end{equation}
We must prove that the left hand side in \eqref{s:in} is dominated by $M\babs{\,z\,}^{\,1-k}\Lap Q_0(z)$ for all large $|\,z\,|$, where $M$ is a suitable constant. If this is false, there are two possibilities. Either $R_0(z)\le (1-M|\,z\,|^{\,1-k})\Lap Q_0(z)$ for arbitrarily large $|\,z\,|$. Then \eqref{s:in} implies
$$M\babs{\,z\,}^{\,k-1}c(z)\le C_1\sqrt{R_0(z)}+C_2|\,z\,|^{\,k-1}\le (C_1'+C_2)|\,z\,|^{\,k-1},$$
and we reach a contradiction for large enough $M$.

 In the remaining case we have $R_0(z)\ge (1+M|\,z\,|^{\,1-k})\Lap Q_0(z)$. Then \eqref{s:in} gives the estimate $R_0(z)\ge cM^{\,2}|\,z\,|^{\,2k-2}$ for some $c>0$. Since $\Lap Q_0(z) \leq c'|\,z\,|^{\,2k-2}$ for some $c'>0$, we obtain
$$ R_0(z)-\Lap Q_0(z) \geq (cM^{\,2}-c') \babs{\,z\,}^{\,2k-2}.$$
Choosing $M$ large enough, we obtain $R_0(z) \geq C_3 M \babs{\,z\,}^{\,4k-4}$ by \eqref{s:in} again. Repeating the above argument gives $R_0(z) \geq C_p M \babs{\,z\,}^{\,2p}$ for all sufficiently large $|\,z\,|$ for some constant $C_p>0$. On the other hand, we will show that
\begin{equation}\label{u-est}R_0(z)\le C(1+|\,z\,|^{\,4k-2})\end{equation} for all $z$, which will give the desired contradiction. To see this, note that for functions $u\in L_a^2(\mu_0)$, the estimate \eqref{vemm} gives
$$\babs{\,u(z)\,}^2e^{-Q_0(z)}\le C\delta^{-2}e^{C|\,z\,|^{\,2k-1}\delta}\|u\|_{L^2(\mu_0)}^2,\quad (|z|\ge 1,\,0<\delta<1).$$
Taking $\delta=|\,z\,|^{1-2k}$ we obtain $\babs{u(z)}^2\le C|\,z\,|^{4k-2}e^{Q_0(z)}
\|u\|_{L^2(\mu_0)}^2$.
Since
$$ L_0(z,z)= \sup \{\, \babs{\,u(z)\,}^{\,2}\, ;\,  u\in L_a^2(\mu_0),\, \|\, u\, \|_{L^2(\mu_0)} \leq 1\,\},$$
we now obtain the estimate \eqref{u-est}.
\end{proof}

\section{Apriori estimates for the one-point function} \label{Sec5}
In this section, we prove part \ref{r1p} of Theorem \ref{apth}.

As before, we write $Q=Q_0+\re H+Q_1$ for the canonical decomposition of $Q$ at $0$, and we write
$\mu_0$ for the measure $d\mu_0=e^{-Q_0}\, dA$. In this section, the assumption that $0$ is in the bulk
of the droplet will become important.

Our arguments below essentially follow by adaptation of the previous section.

Fix a point $\zeta$ in a small neighbourhood of $0$ with $\babs{\,\zeta\,}\ge r_n$.
We also fix a number $\delta_0=\delta_0(\zeta)\ge \const>0$ with $\delta_0(\zeta)\cdot r_n/|\zeta|$ uniformly small, and a smooth
function $\psi$ with $\psi=1$ in $D(0,\delta_0)$ and $\psi=0$ outside $D(0,2\delta_0)$.
We define a function $\chi_\zeta=\chi_{\zeta,n}$ by
$$\chi_\zeta(\omega)=\psi((\omega-\zeta)/r_n).$$

Let
$A(\eta,\omega)$ be a Hermitian-analytic function in a neighbourhood of $(0,0)$, satisfying $A(\eta,\eta)=Q(\eta)$. We shall essentially apply the definition of the approximating kernel (denoted $L_0^\sharp$ in the preceding section) with "$A_0$'' replaced by "$nA$''. We denote this kernel by $\bfL_{n}^\sharp$, viz.
$$\bfL_{n}^\sharp(\zeta,\eta)=n\,\d_1\dbar_2 A(\zeta,\eta)\cdot e^{\,nA(\zeta,\eta)}.$$
The corresponding "approximate projection'' is defined on suitable functions $u$ by
$$\pi_{n}^\sharp u(\zeta)=\langle u,\bfL_\zeta^\sharp\rangle_{L^2(\mu_n)},\quad d\mu_n=e^{-nQ}\, dA,$$
where, for convenience, we write $\bfL^\sharp_\zeta$ instead of $\bfL^\sharp_{n,\zeta}$.

\begin{lem} \label{mack}
Suppose that $u$ is holomorphic in a neighbourhood of $\zeta$ and $\delta_0(\zeta)\cdot r_n/|\zeta|\le\eps_0$ (small enough). Then there is a constant $C=C(\eps_0)$ such that, when $r_n\le |\,\zeta\,|\le r_n\log n$,
$$\babs{\,u(\zeta)-\pi_{n}^\sharp[\chi_\zeta u](\zeta)\,}\le C(1+(\delta_0r_n)^{-1})\|\,u\,\|_{L^2(\mu_n)}e^{\,nQ(\zeta)/2}.$$
\end{lem}

\begin{proof} It will be sufficient to indicate how the proof of Lemma \ref{uv} is modified in the present setting. We start as earlier, by writing
\begin{align*}\pi_n^\sharp[\chi_\zeta f](\zeta)&=-\int\frac {u(\omega)\chi_\zeta(\omega)F(\zeta,\omega)} {\omega-\zeta}
\dbar_\omega\left[e^{-n(A(\omega,\omega)-A(\zeta,\omega))}\right]\, dA(\omega),
\end{align*}
where
$$F(\zeta,\omega)=\frac{(\omega-\zeta)\d_1\dbar_2 A(\zeta,\omega)}
{\dbar_2 A(\omega,\omega)-\dbar_2 A(\zeta,\omega)}.$$
 Here, we may replace "$A$'' by "$A_0$'' to within negligible terms, for the relevant $\zeta$ and $\omega$. More precisely, Taylor's formula gives that
\begin{align}\label{aests}\dbar_2 A(\omega,\omega)-\dbar_2 A(\zeta,\omega)
&=\Lap Q_0(\omega)(1+O(r_n\log n)) \cdot (\omega-\zeta)+O(\,(\omega-\zeta)^{\,2}\,),\\ \d_1\dbar_2 A(\zeta,\omega)
&=\d_1\dbar_2 A_0(\zeta,\omega)(1+O(r_n\log n)),\end{align}
when $r_n\le|\,\zeta\,|\le r_n\log n$ and $\babs{\,\omega-\zeta\,}\le 2\delta_0 r_n$.

From \eqref{aests} and the form of $F$, we see (as in the proof of Lemma \ref{uv}) that
\begin{equation}\label{viso}F(\zeta,\omega)=1+O(\zeta-\omega),\quad \dbar_2 F(\zeta,\omega)=O(\omega-\zeta).\end{equation}
We continue to write $\pi_n^\sharp u(\zeta)=u(\zeta)+\epsilon_1+\epsilon_2$ where
\begin{align*}\epsilon_1&=\int\frac {u(\omega)\cdot \dbar\chi_\zeta(\omega)\cdot F(\zeta,\omega)}
{\omega-\zeta}e^{-n\left[A(\omega,\omega)-A(\zeta,\omega)\right]}\, dA(\omega),\\
\epsilon_2&=\int\frac {u(\omega)\cdot\chi_\zeta(\omega)\cdot\dbar_2 F(\zeta,\omega)}
{\omega-\zeta}e^{-n(A(\omega,\omega)-A(\zeta,\omega))}\, dA(\omega).
\end{align*}
To estimate $\epsilon_1$ and $\epsilon_2$, we note that there is a positive constant $c$ such that
\begin{equation}\label{taly}e^{-n(Q_0(\omega)/2-\re A_0(\zeta,\omega))}\le Ce^{\,nQ_0(\zeta)/2-cn|\zeta|^{2k-2}|\zeta-\omega|^2},\quad \babs{\,\omega-\zeta\,}\le 2\delta_0 r_n.\end{equation}
See Lemma \ref{abo}.

Inserting the estimates in \eqref{viso} and \eqref{taly}, using also that $\dbar\chi_\zeta(w)=0$ when $\babs{\,\zeta-\omega\,}\le \delta_0 r_n$, we find that if $|\zeta|\ge r_n$
\begin{align*}\babs{\,\epsilon_1\,}e^{-nQ(\zeta)/2}&\le C\delta_0^{-1}r_n^{-1}\int\babs{\,u(\omega)\,}|\,\dbar\chi_\zeta(\omega)\,|e^{-nQ(\omega)/2} \, dA(\omega),\\
\babs{\,\epsilon_2\,} e^{-nQ(\zeta)/2}&\le C\int\chi_\zeta(\omega)\babs{\,u(\omega)\,}e^{-nQ(\omega)/2}e^{-cnr_n^{2k-2}|\zeta-\omega|^2}\, dA(\omega).
\end{align*}

Using the Cauchy-Schwarz inequality, we find now that
$$(\babs{\,\epsilon_1\,}+\babs{\,\epsilon_2\,})e^{-nQ(\zeta)/2}\le C(1+\delta_0^{-1}r_n^{-1})\|\,u\,\|_{L^2(\mu_n)}.$$
The proof is complete.
\end{proof}

Choosing $u(\eta)=\bfk_{n}(\eta,\zeta)$ where $\bfk_n$ is the Bergman kernel for the subspace
$\calP_n$ of $L^2(\mu_n)$, we obtain the following estimate, valid when $r_n \le \babs{\,\zeta\,}\le r_n\log n$:
\begin{equation}\label{hk1:eq}\babs{\,\bfk_n(\zeta,\zeta)-\pi_n\left[\chi_\zeta \bfL_\zeta^\sharp\right](\zeta)\,}
\le Cr_n^{-1}\sqrt{\bfk_n(\zeta,\zeta)}\cdot e^{\,nQ(\zeta)/2}.\end{equation}
Here $\pi_n:L^2(\mu_n)\to\calP_n$ is the orthogonal projection, $\pi_n u(\zeta)=\langle u,\bfk_{n,\zeta}\rangle_{L^2(\mu_n)}$. (Cf. \eqref{l0:cor} for details on the derivation of
equation \eqref{hk1:eq} from Lemma \ref{mack}.)


\begin{lem} \label{bp1} For all $\zeta$ in the annulus
$r_n\le\babs{\,\zeta\,}\le  \log n \cdot r_n$, and for $\delta_0(\zeta)\cdot r_n$ a small enough multiple of $|\zeta|$, we have the estimate
$$\babs{\,\pi_n\left[\chi_\zeta\bfL_\zeta^\sharp\right](\zeta)-n\Lap Q(\zeta)\,e^{\,nQ(\zeta)}\,}\le C \sqrt{n}r_n^{-1}|\,\zeta\,|^{\,k-1}e^{\,nQ(\zeta)}.$$
\end{lem}

\begin{proof}
Let
$u_0=\chi_\zeta \bfL_\zeta^\sharp-\pi_n\left[\chi_\zeta \bfL_\zeta^\sharp\right]$
be the norm-minimal solution in $L^2(\mu_n)$ to the problem $\dbar u_0=\dbar f$ where $f=\chi_\zeta \bfL_\zeta^\sharp$. We will prove that the problem $\dbar u=\dbar f$ has a solution $u$
with $u-f\in \Poly(n)$ and
\begin{equation}\label{wilp}\|\,u\,\|_{L^2(\mu_n)}\le Cn^{-1/2}|\zeta|^{-(k-1)}\left\|\,\dbar\left[\chi_\zeta \bfL_\zeta^\sharp\right]\,\right\|_{L^2(\mu_n)}.\end{equation}
This is done by a standard device, which now we briefly recall.

Let $\check{Q}$ be the "obstacle function'' pertaining to $Q$. The main facts about this function to be used here are the following (cf. \cite{ST} for details). The obstacle function can be defined as $\check{Q}=\gamma-2U^\sigma$ where $U^\sigma$ is the logarithmic potential of the equilibrium measure and $\gamma$ is a constant chosen so that $\check{Q}=Q$ on $S$. One has that $\check{Q}$ is harmonic outside $S$, and that its gradient is Lipschitz continuous on $\C$. Furthermore, $\check{Q}(\omega)$ grows like $2\log\babs{\,\omega\,}+O(1)$ as $\omega\to\infty$.

We use the obstacle function to form the strictly subharmonic function
$\phi(\omega)=\check{Q}(\omega)+n^{-1}\log(1+\babs{\,\omega\,}^{\,2})$, and we go on to define a measure
$\mu_n'$ by $d\mu_n'(\omega)=e^{-n\phi(\omega)}\, dA(\omega)$. Write $\calP_n'$ for the subspace
of $L^2(\mu_n')$ of holomorphic polynomials of degree at most $n-1$, and let $\pi_n'$ be the corresponding orthogonal projection. Finally, we put
$$v_0=f-\pi_n'f.$$

Since $\phi$ is now strictly subharmonic, the standard Hörmander estimate can be applied. It gives
$$\|\,v_0\,\|_{L^2(\mu_n')}^{\,2}\le\int_\C\babs{\,\dbar f\,}^{\,2}\frac {e^{-n\phi}}{n\Lap\phi}\, dA.$$
Since $\chi_\zeta$ is supported in the disk $D(\zeta,2\delta_0 r_n)$, and since $\Lap\check{Q}=\Lap Q=\Lap Q_0\cdot (1+o(1))$ there, we see that
$$\|\,v_0\,\|_{L^2(\mu_n')}\le Cn^{-1/2}|\zeta|^{-(k-1)}\left\|\,\dbar f\,\right\|_{L^2(\mu_n)}.$$
Next we use the estimate $n\phi\le nQ+\const$ which holds by the growth assumption on $Q$ near infinity. This gives $\|\,v_0\,\|_{L^2(\mu_n)}\le C\|\,v_0\,\|_{L^2(\mu_n')}$, and so we have shown \eqref{wilp} with $u=v_0$.

Since $n\phi(\omega)=(n+1)\log\babs{\,\omega\,}^{\,2}+O(1)$ as $\omega\to\infty$ we have that $L^2_a(\mu_n')=\Poly(n)$. Hence $u=v_0$ solves, in addition to \eqref{wilp}, the problem
$$\dbar u=\dbar f\quad \text{and}\quad u-f\in\Poly(n).$$

Using the form of $\dbar f=\dbar \chi_\zeta\cdot\bfL_\zeta^\sharp$ and the estimate \eqref{taly}, we find that for $|\omega-\zeta|\le \delta_0 r_n$,
\begin{align*}|\,\dbar u(\omega)\,|^{\,2}e^{-nQ(\omega)}&\le C(n\Lap Q_0(\zeta))^{\,2}|\,\dbar\chi_\zeta(\omega)\,|^{\,2}e^{nQ(\zeta)-2nc|\zeta|^{2k-2}\,\babs{\,\omega-\zeta\,}^{\,2}}.
\end{align*}
By the homogeneity of $\Lap Q_0$ and the fact that $\dbar\chi_\zeta=0$ when $|\omega-\zeta|\le\delta_0r_n$ this gives the estimate
$$\|\,\dbar f\,\|_{L^2(\mu_n)}\le Cn\babs{\,\zeta\,}^{\,2k-2}e^{nQ(\zeta)/2}e^{-cn|\zeta|^{2k-2}(\delta_0r_n)^2}.$$
Applying \eqref{wilp}, we now get
\begin{equation}\label{moww}\|\,u\,\|_{L^2(\mu_n)}\le C\sqrt{n}|\,\zeta\,|^{\,k-1}e^{nQ(\zeta)/2}e^{-cn(\delta_0r_n)^2|\zeta|^{2k-2}}.\end{equation}
We now pick a small constant $\delta$ (independent of $n$) and use
the pointwise-$L^2$ estimate
\begin{align*}\babs{\,u(\zeta)\,}^{\,2}e^{-nQ(\zeta)}&\le C(r_n\delta)^{-2}e^{\,c'nr_n\delta|\zeta|^{2k-1}}\|\,u\,\|_{nQ}^{\,2}.
\end{align*}
Choosing $\delta_0r_n$ as a small multiple of $|\zeta|$ and then $\delta$ small enough, can now use \eqref{moww} to deduce that
$$\babs{\,u(\zeta)\,}e^{-nQ(\zeta)/2}\le Cr_n^{-1}\sqrt{n}|\zeta|^{k-1}e^{nQ(\zeta)/2},$$
finishing the proof.
\end{proof}

\begin{proof}[Proof of Theorem \ref{apth}, part \ref{r1p}.] Fix $\eps>0$ and take
$\zeta$ with $r_n\le \babs{\,\zeta\,}\le\log n\cdot r_n$. By the estimate \eqref{hk1:eq} and Lemma \ref{bp1} we have for all large $n$ that
\begin{align*}\babs{\,\bfR_n(\zeta)-n\Lap Q_0(\zeta)\,}&\le C_1r_n^{-1}\sqrt{\bfR_n(\zeta)}+C_2 r_n^{-k-1}|\zeta|^{k-1},
\end{align*}
for some constants $C_1,C_2$.
Multiplying through by $r_n^{\,2}$ and writing $R_n(z)=r_n^{\,2}\,\bfR_n(\zeta)$, $z=r_n^{-1}\zeta$, we get
\begin{equation}\label{fup}\babs{\,R_n(z)-\Lap Q_0(z)\,}\le C_1\sqrt{R_n(z)}+C_2|z|^{k-1}.\end{equation}
It follows that each limiting $1$-point function $R$ must satisfy
\begin{equation}\label{woul}\babs{\,R(z)-c(z)\babs{\,z\,}^{\,2k-2}\,}\le C_1\sqrt{R(z)}+C_2|z|^{k-1},\quad |z|\ge 1.\end{equation}
where $c(z)=\Lap Q_0(z/|\,z\,|)>0$. The proof of part \ref{r0p} of Theorem \ref{apth} shows that
this is only possible if $R(z)=\Lap Q_0(z)(1+O(z^{1-k}))$ as $z\to\infty$.
\end{proof}

\end{document}